\documentclass[preprint,12pt,authoryear]{elsarticle}

\usepackage{amssymb, amsmath, amsthm, enumitem, kantlipsum, float, verbatim, subcaption, mathrsfs, bm, cancel}

\usepackage[dvipsnames]{xcolor}
\usepackage[all]{xy} 
\usepackage{tikz}
\usetikzlibrary{trees}

\newcommand{\indep}{\perp \! \! \! \perp}

\newtheorem{theorem}{Theorem}[section]

\newtheorem{lemma}[theorem]{Lemma}
\newtheorem{example}{Example}

\theoremstyle{definition}
\newtheorem{definition}{Definition}[section]

\makeatletter
\newcommand*\bigcdot{\mathpalette\bigcdot@{.5}}
\newcommand*\bigcdot@[2]{\mathbin{\vcenter{\hbox{\scalebox{#2}{$\m@th#1\bullet$}}}}}
\makeatother

\journal{IJAR}

\begin{document}

\begin{frontmatter}



\title{Score Equivalence for Staged Trees}


\author{Conor Hughes\corref{cor1}\fnref{inst1}}
\cortext[cor1]{conor.hughes@warwick.ac.uk}

\affiliation[inst1]{organization={Statistics Department},
            addressline={University of Warwick}, 
            city={Coventry},
            postcode={CV4 7AL},
            country={United Kingdom}}

\author[inst2,inst3]{Peter Strong}
\author[inst1]{Aditi Shenvi}

\affiliation[inst2]{organization={Centre for Complexity Science},
            addressline={University of Warwick}, 
            city={Coventry},
            postcode={CV4 7AL},
            country={United Kingdom}}

\affiliation[inst3]{organization={The Alan Turing Institute},
            addressline={96 Euston Road}, 
            city={London},
            postcode={NW1 2DB},
            country={United Kingdom}}

\begin{abstract}


Staged trees are a recently-developed, powerful family of probabilistic graphical models. An equivalence class of staged trees has now been characterised, and two fundamental statistical operators have been defined to traverse the equivalence class of a given staged tree. Here, two staged trees are said to be \textit{statistically equivalent} when they represent the same set of distributions. Probabilistic graphical models such as staged trees are increasingly being used for causal analyses. Staged trees which are within the same equivalence class can encode very different causal hypotheses but data alone cannot help us distinguish between these. Therefore, in using score-based methods to learn the model structure and distributions from data for causal analyses, we should expect that a suitable scoring function is one which assigns the same score to statistically equivalent models. No scoring function has yet been proven to have this desirable property for staged trees. In this paper, we present a novel Bayesian Dirichlet scoring function based on path uniformity and mass conversation, and prove that this new scoring function is score-equivalent for staged trees.

\end{abstract}


\end{frontmatter}


\section{Introduction}
\label{sec:intro}

Probabilistic graphical models (PGMs) combine together a statistical model that decomposes a complex system into a collection of conditional independence relationships among its defining random variable or events, and a graph which provides a visual representation of this decomposition. Staged trees are a recent family of event-based PGMs (i.e. their building blocks are events\footnote{An event is an element or a subset of elements of the state space of a variable.} rather than random variables) that have been shown to be very powerful, especially for modelling processes which have asymmetries in their evolution. These asymmetries include asymmetric conditional independence relations which are also known as \textit{context-specific conditional independencies} \citep{zhang1999role, collazo18}, and asymmetric structures which arise due to the presence of structural zeros and structural missing values within the process \citep{shenvi2020constructing}. The latter form of asymmetry results in processes whose event spaces do not conform to natural product space structures. In other words, these processes cannot be efficiently described using random variables and instead benefit from the granularity offered by event-based models \citep{thwaites10, shenvi21}. Thus, staged trees generalise and often outperform the popular variable-based PGM family of Bayesian networks by not only being able to effectively model processes without structural asymmetries (using the class of \textit{stratified} staged trees) but also processes with structural asymmetries (using the class of \textit{non-stratified} staged trees); see comparisons in \citet{barclay13} for stratified staged trees and \citet{shenvi2018modelling} for non-stratified staged trees. 

\citet{gorgen18} presented a complete characterisation of the statistical equivalence class of staged trees. We say that two staged tree models are statistically equivalent when they encode the same set of probability distributions. This also implies that they encode the same set of conditional independencies. \citet{gorgen18} also presented two fundamental statistical operators, called the \textit{swap} and \textit{resize} operators, that can be used to traverse the equivalence class of any given staged tree. 

PGMs such as staged trees are increasingly being used for causal analyses (see e.g. \citet{shafer96,thwaites10,cowell14,yu2021causal}). Whilst PGMs are built to leverage conditional independencies, in causal analyses we are generally interested in the directionality of the influence in the dependencies. However, it is essential to note that data alone cannot distinguish between PGMs that encode the same conditional independencies but different causal hypotheses \citep{koller2007graphical}. Therefore, in using score-based algorithms for learning the structure and distributions of a staged tree, we would expect that the scoring function also does not distinguish between models within the same statistical equivalence class given a particular dataset. This is an extremely crucial and desirable property for a scoring function, especially within the context of causal analyses, as it mitigates against spurious causal findings. 

Within a Bayesian framework, the score equivalence property additionally needs to account for the prior. In the case where priors are set such that they distinguish between the models within the same equivalence class, the posterior scores are likely to also be able to distinguish between these models. However, in cases where no such distinguishing prior information is available, we typically want to set the priors in a default way such that it does not discriminate between statistical equivalence class. For Bayesian networks, scores such as the Bayesian Dirichlet equivalent uniform (BDeu) score have been proven to have this score equivalence property \citep{chickering95}. The BDeu comes from the Bayesian Dirichlet family of scoring functions, and in particular, the `equivalent uniform' part of its name defines how this scoring function sets the hyperparameters of the priors. \citet{cowell14} presented a na\"ive extension of the BDeu for staged trees; however, this turns out not to have the property of score equivalence for staged trees. In this paper, we present a new score-equivalent BD scoring metric for staged trees.

We note here that staged tree models are statistically equivalent to chain event graph (CEG) models \citep{smith08, collazo18} which are a related PGM family. Every staged tree can be represented as a CEG and in fact, the mapping from a staged tree to its associated CEG is bijective \citep{shenvi2020constructing}. The key difference between the two is that a CEG provides a compact representation of the staged tree and it has fewer nodes and edges than its staged tree analogue. All the results in this paper extend directly to CEGs. 

The rest of this paper is organised as follows. Section \ref{sec:StagedTrees} reviews staged trees and associated concepts. Section \ref{sec:BDeu} presents our novel BDepu scoring function and proves that it is score-equivalent for staged trees. Section \ref{sec:example} illustrates the practicality of the BDepu in analysing a real-world dataset. Finally, we conclude with a discussion in Section \ref{sec:discussion}.

\section{Staged Trees}
\label{sec:StagedTrees}



Let $\mathcal{T}$ be an event tree, a directed tree graph with a series of edges coming out of a root node representing the unfolding of events. Let $\bm{\theta_{\mathcal{T}}}$ be the collection of conditional transition probability vectors for its nodes. We define the \textit{level} of a node as the number of edges it is away from the root node. A \textit{leaf} $l$ is any node with no outgoing edges and a \textit{situation} $s$ is any non-leaf node. The \emph{floret} of a situation $s$, $F(s)$, is the subgraph of $\mathcal{T}$ induced by $s$ and its children. A \textit{stage} $u$ is defined as a set of situations with equivalent conditional transition probability vectors, and we let the stage structure $\mathcal{U}$ be the collection of stages, which partitions the node set of $\mathcal{T}$. Let $\mathcal{S} = (\mathcal{T},\mathcal{U})$ be a staged tree graph, which is the event tree $\mathcal{T}$ with nodes coloured according to their stage. Therefore, two nodes with the same colour are in the same stage. We define a staged tree model as $(\mathcal{S},\bm{\theta_{\mathcal{S}}})$ where $\mathcal{S}$ is the underlying staged tree graph and $\bm{\theta_{\mathcal{S}}}$ is the collection of conditional transition probability vectors for the stages. We will often use $\mathcal{S}$ and/or staged tree to refer to both the staged tree graph and the model itself.  

Let $\Lambda(\mathcal{S})$ be the set of all root-to-leaf paths in $\mathcal{S}$. For a path $\lambda \in \Lambda(\mathcal{S})$, let $E(\lambda)$ be the edge set of that path. For a situation $s_j$ we denote its $k$th outgoing edge as $e_{jk}$. Define the \textit{path set} of situation $s_j$ as $\Lambda(s_j) = \{\lambda \in \Lambda(\mathcal{S}): \exists \ e_{jk} \in E(\lambda), \text{for some k}\}$. For an edge $e_{jk}$, define the path set as $\Lambda(e_{jk}) = \{\lambda \in \Lambda(\mathcal{S}): \ e_{jk} \in E(\lambda), \text{for some k}\}$.

Given a vector $\mathcal{X}=\{X_1,X_2,\dots,X_n\}$ of variables we define\\ $\mathcal{X}^k=\{X_{1},X_{2},\dots,X_{k}\}$ for $1 \leq k \leq n$ and the state space of variable $X_i$ as $\mathbb{X}_i$.

\begin{definition}[$\mathcal{X}$-compatible]
An event tree $\mathcal{S}$ is $\mathcal{X}$-compatible if its node set $V (\mathcal{S})$ consists of a root node $v_0$ together with a node $v(x^k)$ for each $x^k= (x_{1}, x_{2}, . . . , x_{k})$ where $x_{i}\in \mathbb{X}_{i}$ and $1 \leq k \leq n$.
\end{definition}

\begin{definition}[$\mathcal{X}$-stratified]
    A staged tree is said to be an $\mathcal{X}$ stratified staged tree when its underlying event tree is $\mathcal{X}$-compatible.
\end{definition}

In this paper, we define a staged tree as stratified (as in \cite{cowell14} and \cite{shenvi21}) if it is $\mathcal{X}$-stratified for some $\mathcal{X}$. Note that an alternative definition of stratified staged trees has been used, for details see page 62 of \citet{collazo18}.

In \citet{gorgen18}, and elaborated in \citet{gorgen21}, it is proven that the statistical equivalence class of a staged tree can be traversed by two operators: the \emph{swap} operator, and the \emph{resize} operator. In this paper, just as in \citet{gorgen18}, we restrict our attention to square-free staged trees, in which no two situations on the same root-to-leaf path are in the same stage. \par

The swap operates on a particular type of subtree called a \emph{twin}. We will refer to any leaf in a subtree as a \emph{subleaf}. 

\begin{definition}[Twin \citep{gorgen18}]
A twin around some stage $u$ is the probability subtree $(\mathcal{S}_u,\bm{\theta_{\mathcal{S}_u}}) \subseteq (\mathcal{S},\bm{\theta_\mathcal{S}})$ where all root-to-subleaf paths have exactly two edges, and each child of the root is in the same stage $u$.

\end{definition}


The swap operator reorders the situations in a twin. In Figure \ref{fig:swap_example}, there are two staged trees, with the colours indicating the staging, and the edges labelled with the transition probability parameters. Note that these two staged trees are non-stratified. On the left hand side is $\mathcal{S}$, which has a twin around the stage $u_2 = \{s_1, s_2\}$. On the right hand side is $\mathcal{S}'$, which is the staged tree after a swap is applied to the twin around $u_2$. Note how the situation $s_0'$ is in the same stage as $s_1$ and $s_2$, while $s_1', s_2'$ and $s_3'$ are in the same stage as $s_0$. We point out how the swap repositions some of the situations and thus stages throughout the tree, but ultimately maintains the same set of parameters. The paths in both trees are identical, up to a change of ordering of the events, and so there is a one to one correspondence between the leaves in both trees.

\begin{figure}[!ht]
\centering
\begin{minipage}{0.5 \textwidth}
\centering
\scalebox{0.4}{  
\xymatrixcolsep{5.5pc} \xymatrix{
&&\text{\huge $\mathcal{S}$}&&\\
 &&&\text{\huge $l_1$}\\
 &&\text{\huge  \color{ForestGreen} $s_{3}$} \ar[r]|-{\txt{\huge $\theta_{32}$}}
			        \ar[ur]|-{\txt{\huge $\theta_{31}$}}
 &\text{\huge  $l_2$}\\
  &&&\text{\huge $l_3$}\\
   &\text{\huge \color{red} $s_{1}$} \ar[r]|-{\txt{\huge  $\theta_{22}$}}
			       \ar[uur]|-{\txt{\huge  $\theta_{21}$}}
                      \ar[dr]|-{\txt{\huge $\theta_{23}$}}
&\text{\huge  \color{blue} $s_{4}$}\ar[ur]|-{\txt{\huge  $\theta_{41}$}} 
\ar[r]|-{\txt{\huge  $\theta_{42}$}}&\text{\huge $l_4$}\\
 &&\text{\huge  \color{CarnationPink} $s_{5}$}\ar[r]|-{\txt{\huge  $\theta_{51}$}} 
\ar[dr]|-{\txt{\huge  $\theta_{52}$}}&\text{\huge $l_5$} \\
 &&&\text{\huge  $l_6$}\\
\text{\huge \color{orange} $s_{0}$} \ar[uuur]|-{\txt{\huge  $\theta_{11}$}}
			        \ar[dddr]|-{\txt{\huge $\theta_{12}$}}&&&&\\
&&&\text{\huge  $l_7$} \\
&&\text{\huge \color{ForestGreen} $s_{6}$} \ar[ur]|-{\txt{\huge  $\theta_{31}$}}
			        \ar[r]|-{\txt{\huge $\theta_{32}$}}
 &\text{\huge $l_8$}\\
  &\text{\huge \color{red} $s_{2}$} \ar[ur]|-{\txt{\huge $\theta_{21}$}}
                      \ar[r]|-{\txt{\huge $\theta_{22}$}}
			       \ar[ddr]|-{\txt{\huge $\theta_{23}$}}
          &\text{\huge \color{blue} $s_{7}$}\ar[dr]|-{\txt{\huge  $\theta_{42}$}}
           \ar[r]|-{\txt{\huge  $\theta_{41}$}}&\text{\huge $l_9$}\\
   &&&\text{\huge $l_{10}$}\\
 &&\text{\huge $l_{11}$}&\\
} } 
\end{minipage}\hfill
\begin{minipage}{0.5 \textwidth}
\centering
\scalebox{0.4}{  
\xymatrixcolsep{5.5pc} \xymatrix{
&&\text{\huge $\mathcal{S}'$}&&\\
 &&&\text{\huge $l_1'$}\\
 &&\text{\huge  \color{ForestGreen} $s_{4}'$} \ar[r]|-{\txt{\huge  $\theta_{32}$}}
			        \ar[ur]|-{\txt{\huge  $\theta_{31}$}}
 &\text{\huge $l_2'$}\\
  &\text{\huge \color{orange} $s_{1}'$} \ar[r]|-{\txt{\huge $\theta_{12}$}}
			       \ar[ur]|-{\txt{\huge $\theta_{11}$}}
          &\text{\huge  \color{ForestGreen} $s_{5}'$} \ar[r]|-{\txt{\huge $\theta_{31}$}}
\ar[dr]|-{\txt{\huge $\theta_{32}$}}
          &\text{\huge $l_3'$}\\
   &&&\text{\huge $l_4'$}\\
 &&\text{\huge  \color{blue} $s_{6}'$} \ar[r]|-{\txt{\huge  $\theta_{41}$}}
\ar[dr]|-{\txt{\huge $\theta_{42}$}}&\text{\huge  $l_5'$} \\
 &&&\text{\huge  $l_6'$}\\
\text{\huge \color{red} $s_{0}'$} \ar[uuuur]|-{\txt{\huge $\theta_{21}$}}
                    \ar[r]|-{\txt{\huge  $\theta_{22}$}}
			        \ar[ddddr]|-{\txt{\huge  $\theta_{23}$}}
&\text{\huge \color{orange} $s_{2}'$} \ar[uur]|-{\txt{\huge  $\theta_{11}$}}
			       \ar[ddr]|-{\txt{\huge  $\theta_{12}$}}&&&\\
&&&\text{\huge $l_7'$} \\
&&\text{\huge \color{blue} $s_{7}'$} \ar[ur]|-{\txt{\huge $\theta_{41}$}}
			        \ar[r]|-{\txt{\huge $\theta_{42}$}}
 &\text{\huge  $l_8'$}\\
 &&&\text{\huge $l_9'$} \\
 &\text{\huge \color{orange} $s_{3}'$} \ar[r]|-{\txt{\huge  $\theta_{11}$}}
			       \ar[dr]|-{\txt{\huge  $\theta_{12}$}}
          &\text{\huge \color{CarnationPink} $s_{8}'$} \ar[r]|-{\txt{\huge $\theta_{52}$}}
			        \ar[ur]|-{\txt{\huge  $\theta_{51}$}}
          &\text{\huge $l_{10}'$}\\
 &&\text{\huge $l_{11}'$}
 &\\
} } 
\end{minipage}
\caption{\label{fig:swap_example}The staged trees before and after a swap on the twin around $u_2 = \{s_1, s_2\}$.}
\end{figure}
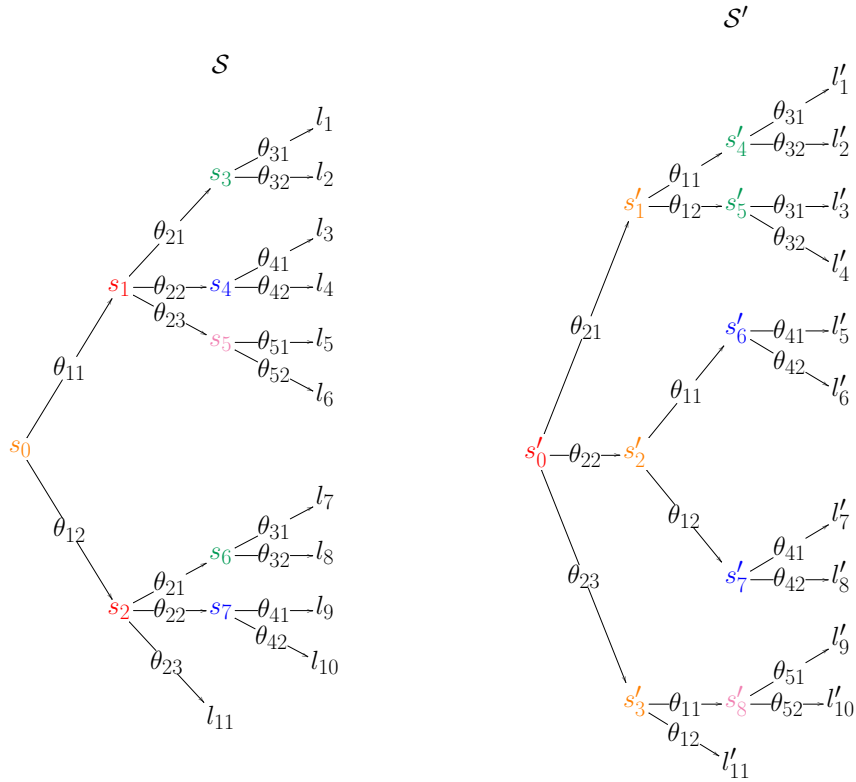

 

The resize operator either contracts subtrees rooted at some node $s$ into a floret, or expands a floret into a larger subtree. In such a floret, each edge corresponds to a path in the subtree and has an associated transition probability given by traversing the root-to-subleaf path in the subtree. This probability is the product of the transition probabilities for each constituent edge of the root-to-subleaf path in the subtree.\par

Based on the literature, there are two cases where the resize operator produces a staged tree \citep{gorgen18}. The first is if the subtree is \emph{saturated}; i.e. every situation in the subtree is in a stage by itself, and so the colour of each situation in the subtree is unique in the tree as a whole. Hence, as the subtree contains no additional parametric information, each root-to-subleaf path can be contracted into a single edge emanating from the root and entering the subleaf. The second case is resizing multiple subtrees that are \emph{conditionally saturated}. That is, each of these subtrees are identical, while the stage colours within each subtree are unique and do not repeat anywhere else outside the subtrees. This case is a generalisation of the first.

We again look at the non-stratified staged tree $\mathcal{S}$ from Figure \ref{fig:swap_example}. The subtree rooted at $s_1$ and containing the paths to $l_1$ to $l_4$ is the same as the subtree rooted at $s_2$ and containing the paths to $l_7$ to $l_{10}$. Meanwhile, the colours in these subtrees do not repeat anywhere inside the subtrees themselves, nor outside them, and so they are conditionally saturated. Hence, the subtrees can be resized as long as we resize both. There are four paths in each subtree, and so the resultant florets will have four edges, with conditional transition probability equal to the product of the transition probabilities along the edges in the subtree. The remaining edges, to $s_5$ and $l_{11}$, will not be changed. The staged tree after applying the resize operator to these subtrees, $\mathcal{S}''$, can be seen alongside $\mathcal{S}$ in Figure \ref{fig:resize_example}.

\begin{figure}[!ht]
\centering
\begin{minipage}{0.5 \textwidth}
\centering
\scalebox{0.4}{  
\xymatrixcolsep{5.5pc} \xymatrix{
&&\text{\huge $\mathcal{S}$}&&\\
 &&&\text{\huge $l_1$}\\
 &&\text{\huge  \color{ForestGreen} $s_{3}$} \ar[r]|-{\txt{\huge $\theta_{32}$}}
			        \ar[ur]|-{\txt{\huge $\theta_{31}$}}
 &\text{\huge  $l_2$}\\
  &&&\text{\huge $l_3$}\\
   &\text{\huge \color{red} $s_{1}$} \ar[r]|-{\txt{\huge  $\theta_{22}$}}
			       \ar[uur]|-{\txt{\huge  $\theta_{21}$}}
                      \ar[dr]|-{\txt{\huge $\theta_{23}$}}
&\text{\huge  \color{blue} $s_{4}$}\ar[ur]|-{\txt{\huge  $\theta_{41}$}} 
\ar[r]|-{\txt{\huge  $\theta_{42}$}}&\text{\huge $l_4$}\\
 &&\text{\huge  \color{CarnationPink} $s_{5}$}\ar[r]|-{\txt{\huge  $\theta_{51}$}} 
\ar[dr]|-{\txt{\huge  $\theta_{52}$}}&\text{\huge $l_5$} \\
 &&&\text{\huge  $l_6$}\\
\text{\huge \color{orange} $s_{0}$} \ar[uuur]|-{\txt{\huge  $\theta_{11}$}}
			        \ar[dddr]|-{\txt{\huge $\theta_{12}$}}&&&&\\
&&&\text{\huge  $l_7$} \\
&&\text{\huge \color{ForestGreen} $s_{6}$} \ar[ur]|-{\txt{\huge  $\theta_{31}$}}
			        \ar[r]|-{\txt{\huge $\theta_{32}$}}
 &\text{\huge $l_8$}\\
  &\text{\huge \color{red} $s_{2}$} \ar[ur]|-{\txt{\huge $\theta_{21}$}}
                      \ar[r]|-{\txt{\huge $\theta_{22}$}}
			       \ar[ddrr]|-{\txt{\huge $\theta_{23}$}}
          &\text{\huge \color{blue} $s_{7}$}\ar[dr]|-{\txt{\huge  $\theta_{42}$}}
           \ar[r]|-{\txt{\huge  $\theta_{41}$}}&\text{\huge $l_9$}\\
   &&&\text{\huge $l_{10}$}\\
 &&&\text{\huge $l_{11}$}&\\
} } 
\end{minipage}\hfill
\begin{minipage}{0.5 \textwidth}
\centering
\scalebox{0.4}{  
\xymatrixcolsep{5.5pc} \xymatrix{
&&\text{\huge $\mathcal{S}''$}&&\\
 &&&\text{\huge $l_1$}\\
 &&&\text{\huge  $l_2$}\\
 &&&\text{\huge  $l_2$}\\
 &\text{\huge \color{gray} $s_{1}'$} \ar[uuurr]|-{\txt{\huge $\theta_{21}\theta_{31}$}}
 \ar[uurr]|-{\txt{\huge  $\theta_{21}\theta_{32}$}}
 \ar[urr]|-{\txt{\huge  $\theta_{22}\theta_{41}$}}
 \ar[rr]|-{\txt{\huge  $\theta_{22}\theta_{42}$}}
 \ar[dr]|-{\txt{\huge   $\theta_{23}$}}
&&\text{\huge $l_4$} \\
 &&\text{\huge  \color{CarnationPink} $s_{5}$}\ar[r]|-{\txt{\huge  $\theta_{51}$}} 
\ar[dr]|-{\txt{\huge  $\theta_{52}$}}&\text{\huge  $l_5$}\\
 &&&\text{\huge  $l_6$}\\
\text{\huge \color{orange} $s_{0}$} \ar[uuur]|-{\txt{\huge  $\theta_{11}$}}
			        \ar[dddr]|-{\txt{\huge $\theta_{12}$}}&&&&\\
&&&\text{\huge  $l_7$} \\
&&&\text{\huge  $l_8$} \\
&\text{\huge \color{gray} $s_{2}'$} \ar[uurr]|-{\txt{\huge  $\theta_{21}\theta_{31}$}}
\ar[urr]|-{\txt{\huge  $\theta_{21}\theta_{32}$}}
\ar[rr]|-{\txt{\huge  $\theta_{22}\theta_{41}$}}
\ar[drr]|-{\txt{\huge  $\theta_{22}\theta_{42}$}}
\ar[ddrr]|-{\txt{\huge  $\theta_{23}$}}
&&\text{\huge $l_9$}\\
 &&&\text{\huge $l_{10}$}\\
 &&&\text{\huge $l_{11}$} \\
} } 
\end{minipage}
\caption{\label{fig:resize_example}The staged trees before and after the conditionally saturated subtrees rooted at $s_1$ and $s_2$ are resized.}
\end{figure}
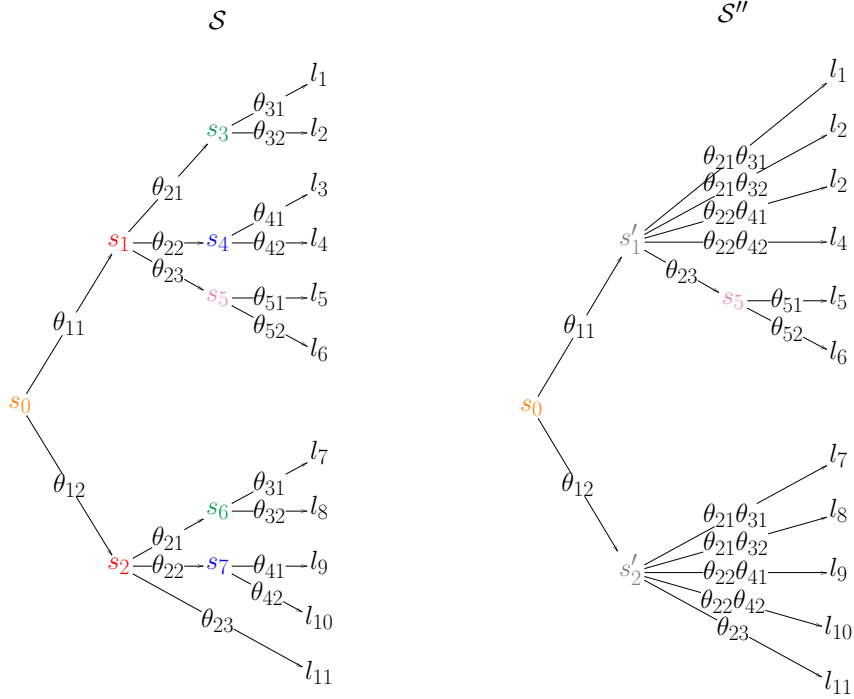


For formal definitions of the swap and resize operators please see \citet{gorgen18} and \citet{gorgen21}. 

\section{Score Equivalence for Staged Trees}
\label{sec:BDeu}


\subsection{BD-Metric} \label{subsec:motivation}

Model selection in staged trees entails clustering of the nodes in its underlying event tree into stages. All current model selection algorithms for staged trees in the literature are score-based and aim to maximise a chosen score function. The most popular of these algorithms is the agglomerative hierarchical clustering (AHC) algorithm \citep{freemansmith11} which is typically paired with the BD-metric as the choice of score function. 

Since the staged tree model aims to reduce the parameter space by leveraging symmetries within its event tree model, any score function for the staged tree will involve calculations over its stages rather than its situations. Consider a staged tree with $J$ stages denoted as $u_j$, for $j = 1, 2, \ldots, J$. Also consider a complete random sample $\mathcal{D} = \{\bm{n}_1, \bm{n}_2, \ldots, \bm{n}_{J}\}$ such that $\bm{n}_j = (n_{j1}, n_{j2}, \ldots, n_{jr_j})$ is the vector of the number of units in the sample that arrive at stage $u_j$ and traverse along one of its $r_j$ edges. \citet{freemansmith11} proved that under the assumptions of complete random sampling and where the priors for the stage parameters $\pmb{\theta}_j$ are independent \textit{a priori} with $\pmb{\theta}_j \sim Dirichlet(\alpha_{j1}, \alpha_{j2}, \ldots, \alpha_{jr_j})$, the stage posteriors follow $Dirichlet(\alpha_{j1} + n_{j1}, \alpha_{j2} + n_{j2}, \ldots, \alpha_{jr_j} + n_{jr_j})$. We also define $\overline{\alpha}_j = \sum_{k=1}^{r_j}\alpha_{jk}$ and $\overline{n}_j = \sum_{k=1}^{r_j}n_{jk}$. Under these assumptions, \citet{barclay13} gives the BD-metric for a staged tree $\mathcal{S}$ with data $\mathcal{D}$ and prior specification $\bm{\alpha}$ as:
\begin{align}
    \text{BD}(\mathcal{S},\mathcal{D};\bm{\alpha}) = \prod_{j=1}^{J}\left[ \frac{\Gamma(\overline{\alpha}_{j})}{\Gamma(\overline{\alpha}_{j}+\overline{n}_{j})} \prod_{k=1}^{r_j} \frac{\Gamma(\alpha_{jk}+n_{jk})}{\Gamma(\alpha_{jk})} \right].
    \label{eqn:bd_score_trees}
\end{align}
\noindent where $\bm{\alpha} = (\bm{\alpha_1},\dots,\bm{\alpha_J})$ and $\bm{\alpha_j} = (\alpha_{j1},\dots,\alpha_{jr_j}).$\par
The BD-metric for BNs takes a similar form to that in Equation \eqref{eqn:bd_score_trees} where the way in which the $\alpha_{jk}$ are set leads to different types of scores with different properties, e.g. the BDeu \citep{heckerman95}, K2 \citep{cooper1991bayesian}, BD sparse \citep{scutari16}. This method of deriving different BD scores through the setting of $\alpha_{jk}$'s can also extend to staged trees. 

\subsection{CS-BDeu} \label{subsec:realted_work}

Our focus in this paper is on deriving a scoring function that is score-equivalent, i.e. gives the same score for staged trees in the same statistical equivalence class. Our approach is to do this through a specific setting of $\alpha_{jk}$ in the BD-metric in Equation \eqref{eqn:bd_score_trees}. For BNs, the BDeu score is score-equivalent. \citet{cowell14} proposed a way of setting $\alpha_{jk}$ in the BD-metric for stratified staged trees that they conjectured resulted in a BDeu-analogue for the class of stratified staged trees. We shall call this the CS-BDeu score.

The CS-BDeu was developed for stratified staged trees where each variable has a representative node on every root-to-leaf path and where the nodes representing a particular variable are all at the same distance from the root. By convention, in stratified staged trees, a stage set can only contain situations that are on the same level. As a result, the CS-BDeu was defined to involve calculations over the stages separately within each level. 

Mirroring the approaches of BNs, an \textit{imaginary sample size} $\alpha$ is set at the root and then propagated uniformly through the tree using the property of mass conservation (see pg 112 of \citet{collazo18}). This mass conservation property plays an important role in our proof of score equivalence, as will be shown later.

However, as the CS-BDeu was only defined for stratified staged trees, the setting of the hyperparameters must be modified to be applicable to any staged tree. The modification we have chosen has the mass conservation property by design. 

We extend the notation introduced before to include levels. Let $i = 1, 2, \ldots I$ index the levels, $j = 1, 2, \ldots, J_i$ index the stages in level $i$, and $k = 1, 2, \ldots, r_{ij}$ index the outgoing edges in the $j$th stage of level $i$. The CS-BDeu for some $\alpha$ is given as
\begin{align}
    \text{CS-BDeu}(\mathcal{S},\mathcal{D};\bm{\alpha}) = \prod_{i=1}^{I}\left[ \prod_{j = 1}^{ J_i} \frac{\Gamma(\dot{\alpha}_{ij})}{\Gamma(\dot{\alpha}_{ij}+\overline{n}_{ij})} \prod_{k=1}^{r_{ij}} \frac{\Gamma(\alpha_{ijk}+n_{ijk})}{\Gamma(\alpha_{ijk})} \right]
    \label{eqn:cs_bdeu}
\end{align}
\noindent where $\alpha_{ijk} = \dot{\alpha}_{ij} /r_{ij}$ and $\dot{\alpha}_{ij}$ is defined as the sum of the hyperparameters of the edges that enter the stage $u_{ij}$, with the convention that the hyperparameter for the root is $\alpha$. Note that $\overline{\alpha}_{ij} = \sum_{k = 1}^{r_{ij}} \alpha_{ijk} = \dot{\alpha}_{ij}$, and so the sum of the hyperparameters that enter the stage equal the sum of the hyperparameters exiting the stage, the mass conservation property. 

The CS-BDeu, when the hyperparameters are set this way, is score-equivalent only for the class of stratified staged trees and not for non-stratified stages trees. We prove the latter below by way of a counter-example. Further, observe that, as shown in Example \ref{ex:csbdeu_counter_example}, the equivalence class of a stratified staged tree also contains non-stratified staged trees, as $\mathcal{S}'$ is stratified and $\mathcal{S}$ is not. 

\begin{example}[Counter example for score equivalence of the CS-BDeu]
\label{ex:csbdeu_counter_example}
Consider the staged trees $\mathcal{S}$ and $\mathcal{S}'$ given in Figure \ref{fig:cs_example}. Both these staged trees belong to the same equivalence class, as $\mathcal{S}'$ can be obtained from $\mathcal{S}$ by using the resize operator at $s_0$. We have that $\mathcal{S}$ is non-stratified whilst $\mathcal{S}'$ is stratified. The prior hyperparameters as specified by the CS-BDeu along the edges of $\mathcal{S}$ and $\mathcal{S}'$ are shown in Figure \ref{fig:cs_example} for some $\alpha > 0$. From this, we can see that the hyperparameter calculation is not consistent for the edges entering the leaves between the two trees and therefore, the CS-BDeu scores for $\mathcal{S}$ and $\mathcal{S}'$ are not equivalent as shown below. 
Assume there are three individuals in the data sample, one along each path. Then:
\footnotesize
\begin{align*}
    \text{CS-BDeu}(\mathcal{S},\mathcal{D};\bm{\alpha}) = &\left[\frac{\Gamma(\alpha)}{\Gamma(\alpha+3)}\frac{\Gamma(\tfrac{\alpha}{2}+2)\Gamma(\tfrac{\alpha}{2}+1)}{\Gamma(\tfrac{\alpha}{2})^{2}}\right]\left[\frac{\Gamma(\tfrac{\alpha}{2})}{\Gamma(\tfrac{\alpha}{2}+2)}\frac{\Gamma(\tfrac{\alpha}{4}+1)^2}{\Gamma(\tfrac{\alpha}{4})^2}\right]\\&\left[\frac{\Gamma(\tfrac{\alpha}{2})}{\Gamma(\tfrac{\alpha}{2}+1)}\frac{\Gamma(\tfrac{\alpha}{2}+1)}{\Gamma(\tfrac{\alpha}{2})}\right]\\
    =&\frac{\Gamma(\alpha)}{\Gamma(\alpha+3)}\frac{\Gamma(\tfrac{\alpha}{2}+1)\Gamma(\tfrac{\alpha}{4}+1)^2}{\Gamma(\tfrac{\alpha}{2})\Gamma(\tfrac{\alpha}{4})^2} \\
    \text{CS-BDeu}(\mathcal{S}',\mathcal{D};\bm{\alpha}) = &\frac{\Gamma(\alpha)}{\Gamma(\alpha+3)}\frac{\Gamma(\tfrac{\alpha}{3}+1)^3}{\Gamma(\tfrac{\alpha}{3})^3} \neq \text{CS-BDeu}(\mathcal{S},\mathcal{D};\bm{\alpha})
\end{align*}
 
\normalsize

\end{example}

\begin{figure}[!ht]
\centering
\begin{minipage}{0.5 \textwidth}
\centering
\scalebox{0.4}{  
\xymatrixcolsep{5.5pc} \xymatrix{
&\text{\huge $\mathcal{S}$}&\\
 &&\text{\huge $l_1$}\\
 &\text{\huge  \color{ForestGreen} $s_{1}$} \ar[r]|-<(0.7){\txt{\huge $\tfrac{\alpha}{4}$}}
			        \ar[ur]|-<(0.3){\txt{\huge $\tfrac{\alpha}{4}$}}
 &\text{\huge  $l_2$}\\
   \text{\huge  $s_{0}$} \ar[ur]|-{\txt{\huge  $\tfrac{\alpha}{2}$}}
                      \ar[dr]|-{\txt{\huge $\tfrac{\alpha}{2}$}}&&\\
&\text{\huge \color{blue} $s_{2}$}\ar[r]|-{\txt{\huge  $\tfrac{\alpha}{2}$}} &\text{\huge $l_3$}\\
} } 
\end{minipage}\hfill
\begin{minipage}{0.5 \textwidth}
\centering
\scalebox{0.4}{  
\xymatrixcolsep{5.5pc} \xymatrix{
&\text{\huge $\mathcal{S}'$}&\\
 &&\text{\huge $l_1$}\\
 \text{\huge  $s_{0}$} \ar[urr]|-<(0.3){\txt{\huge  $\tfrac{\alpha}{3}$}}
                        \ar[rr]|-{\txt{\huge  $\tfrac{\alpha}{3}$}}
                      \ar[drr]|-<(0.7){\txt{\huge $\tfrac{\alpha}{3}$}}
 &&\text{\huge  $l_2$}\\
&&\text{\huge $l_3$}\\
} } 
\end{minipage}
\caption{The staged trees $\mathcal{S}$ and $\mathcal{S}'$ labelled with the edge hyperparameters under the CS-BDeu.}
\label{fig:cs_example}
\end{figure}
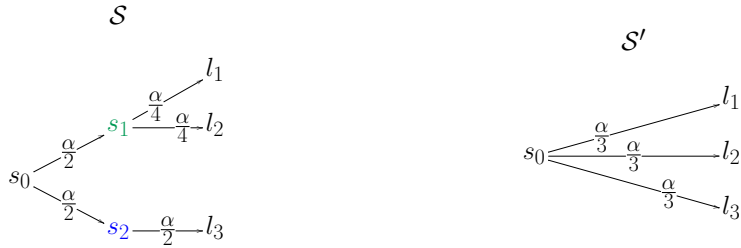

There are other possible options to set the hyperparameters, which we outline in \ref{sec:app:alt_spec}, along with counterexamples for score equivalence. 

\subsection{BDepu} \label{subsec:Bdepu}

We propose an alternative way of setting the $\alpha_{jk}$'s in Equation \eqref{eqn:bd_score_trees} such that it leads to score equivalence for all staged trees, stratified and non-stratified. Note that our proposed BD score is calculated directly over the stages as in Section \ref{subsec:motivation} and therefore, will not involve levels as in Section \ref{subsec:realted_work} and its associated indexing. 

Our approach is centred around the concept of \textit{path uniformity}, where every possible path is assumed equally probable \emph{a priori}. That is, the prior hyperparameters are directly based on the number of paths in the staged tree. Note that \citet{barclayhutton14} discussed the setting of uniform priors on paths albeit not in the context of score equivalence. 

This approach assigns each root-to-leaf path a uniform prior weight and the hyperparameters for each stage and its outgoing edges are based on the number of root-to-leaf paths it is a constituent of. We call this the Bayesian Dirichlet equivalent path uniform (BDepu) prior and it is given by setting $\overline{\alpha_{j}}$ and $\alpha_{jk}$ for $j = 1, 2, \ldots, J$ and $k = 1, 2, \ldots, r_j$ in Equation \eqref{eqn:bd_score_trees} as follows:
\begin{align}
    \alpha_{jk} = \frac{\alpha}{\vert{\Lambda(\mathcal{S})\vert}} \sum_{m = 1}^{h_j} \vert{\Lambda(e_{jk}^{m})\vert} \quad \quad \text{and} \quad \quad \overline{\alpha}_{j} = \sum_{k = 1}^{r_j} \alpha_{jk}
    \label{eqn:BDepu_setting}
\end{align}
\noindent where $h_j$ is the number of situations belonging to stage $u_j$ and $\Lambda(e_{jk}^{m})$ is the path set (see Section \ref{sec:StagedTrees}) of the $k$th outgoing edge of the $m$th situation belonging to stage $u_j$.


Observe that our approach above is in contrast to the CS-BDeu approach of being level-based rather than path-based. In \ref{sec:app:csbdeu} we prove that for stratified staged trees, the CS-BDeu is equivalent to the BDepu. 

Due to the one-to-one correspondence between the leaves and root-to-leaf paths in a staged tree, the BDepu setting is equivalent to assigning an equal prior weight on each leaf and propagating this uniformly backwards through the tree. In this way, we obtain prior hyperparameters for each situation and its outgoing edges. The BD-metric requires hyperparameters over the stages (rather than the situations) which are obtained by summing the hyperparameters for the situations and their outgoing edges which comprise each stage. This is precisely what is done in Statement \eqref{eqn:BDepu_setting}. Thus, the BDepu respects the mass conservation property of \citet{collazo18}.


We will now prove that the BDepu is score-equivalent by proving it is score-equivalent under both the swap and resize operators. 

First, we will assume without generality that the swap operator between two staged trees is acting on only one twin. For brevity, we will also assumed fixed $\alpha$ and $\mathcal{D}$ (up to a relabelling) in each case. 

\setcounter{theorem}{0}
\begin{lemma}
Let $\tau: (\mathcal{S},\theta_{\mathcal{S}}) \rightarrow (\mathcal{S}',\theta_{\mathcal{S}'})$ be a swap operator between staged trees that maps the twin $(\mathcal{S},\theta_{\mathcal{S}})_u$ to $(\mathcal{S}',\theta_{\mathcal{S}'})_u$. Then, BDepu($\mathcal{S}_u$) = BDepu($\mathcal{S}_u'$).\label{lemma1}
\end{lemma}
\begin{proof} 
Let $s_0$ be the root of $\mathcal{S}_u$, which has $r_1$ outgoing edges. As all children of $s_0$ are in the same stage by definition, they must each have an equal number of outgoing edges, in this case $r_2$. The children of $s_0$ can be labelled $s_{0k}, \ k =1,\dots r_1$, while the children of each situation $s_{0k}$ are labelled $s_{kt}, \ t=1,\dots,r_2.$ Each $s_{kt}$ is a subleaf in the twin. Please see Figure \ref{fig:notation_ex} for a demonstration of the notation. Similarly, one can define the outgoing edge from $s_0$ to $s_{0k}$ as $e_{0k}$, and from $s_{0k}$ to $s_{kt}$ as $e_{kt}$. 

\begin{figure}[!ht]
    \centering
    \includegraphics{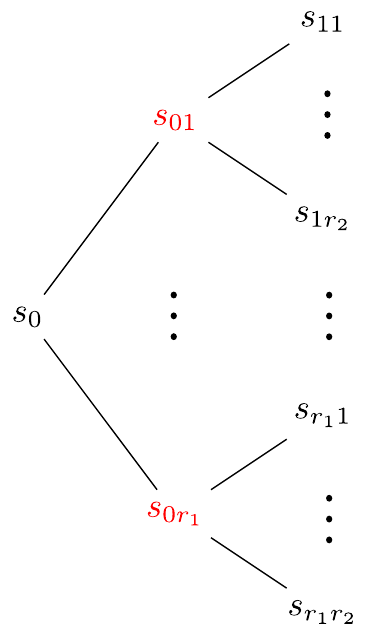}
    \caption{An example of the notation used in Section \ref{subsec:Bdepu}.}
    \label{fig:notation_ex}
\end{figure}

Each edge $e_{kt}, k = 0,1,\dots,r_1, \ t = 1,\dots,r_1$, has a number of descendant leaf nodes in the tree $\mathcal{S}$, which is the cardinality of the path set, $\vert\Lambda(e_{kt})\vert$ = $L_{kt}$. These edges also have corresponding counts in the data, labelled $n_{kt}$. It must be true that for $k \neq 0$
\begin{align}
    L_{0k} = \sum_{t=1}^{r_2}L_{kt}, \ \  
    n_{0k} = \sum_{t=1}^{r_2}n_{kt}\label{eqn:reverse_ind}
\end{align}
because if $e_{kt}$ lies on a root-to-leaf path, so does $e_{0k}$.\par
For $s_0$, we have:
\begin{align*}
    \overline{\alpha}_0 &= \sum_{k=1}^{r_1} \alpha_{0k} = \frac{\alpha}{L}\sum_{k=1}^{r_1} L_{0k} = \frac{\alpha}{L}\sum_{k=1}^{r_1}\sum_{t=1}^{r_2} L_{kt}\\
    \overline{n}_0 &= \sum_{k=1}^{r_1} n_{0k} = \sum_{k=1}^{r_1}\sum_{t=1}^{r_2} n_{kt}
\end{align*}
Each $s_{0k}$ are in the same stage $u$, with $r_2$ outgoing edges. We thus use $\alpha_{ut}, \ t=1,\dots,r_2$, and similarly $n_{ut}$ to represent the hyperparameters and counts for the stage's edges, after summing across all $r_1$ situations. These are:
\begin{align*}
    \alpha_{ut} &= \sum_{k=1}^{r_1} \alpha_{kt} = \frac{\alpha}{L}\sum_{k=1}^{r_1} L_{kt} =  \frac{\alpha}{L} L_{\bigcdot t} = \alpha_{\bigcdot t}\\
    n_{ut} &= \sum_{k=1}^{r_1} n_{kt} = n_{\bigcdot t} 
\end{align*}
Here we use a period ($\bigcdot$) in the subscript to indicate a marginal sum. We point out that $\overline{\alpha}_u = \overline{\alpha}_0$ and $\overline{n}_u = \overline{n}_0$. Thus, the BDepu metric for $\mathcal{S}_u$ is:
\begin{align}
    \text{BDepu}(\mathcal{S}_u) &= \frac{\Gamma(\overline{\alpha}_0)}{\Gamma(\overline{\alpha}_0)+\Gamma(\overline{n}_0)}\left[ \prod_{k=1}^{r_1}\frac{\Gamma(\alpha_{0k} + n_{0k})}{\Gamma(\alpha_{0k})}\right]\frac{\Gamma(\overline{\alpha}_u)}{\Gamma(\overline{\alpha}_u)+\Gamma(\overline{n}_u)}\left[\prod_{t=1}^{r_2}\frac{\Gamma(\alpha_{ut} + n_{ut})}{\Gamma(\alpha_{ut})}\right]\nonumber\\
    &=\left(\frac{\Gamma(\overline{\alpha}_0)}{\Gamma(\overline{\alpha}_0)+\Gamma(\overline{n}_0)}\right)^2\left[ \prod_{k=1}^{r_1}\frac{\Gamma(\alpha_{0k} + n_{0k})}{\Gamma(\alpha_{0k})}\right]\left[\prod_{t=1}^{r_2}\frac{\Gamma(\alpha_{\bigcdot t} + n_{\bigcdot t})}{\Gamma(\alpha_{\bigcdot t})}\right]\label{eqn:bd_swap_t}
\end{align}
In $\mathcal{S'}_u$ the edges are swapped. There are now $r_2$ edges outgoing from the root $s_0'$, with each child having $r_1$ edges with the situations and edges labelled $s_{0k}^{'}$ and $e_{0k}^{'}$ respectively. Similarly, the subleaves and their incoming edges in $\mathcal{S'}_u$ are labelled $s_{kt}^{'}$ and $e_{kt}^{'}$. Each of these edges have counts $n_{0k}^{'}$ or $n_{kt}^{'}$, and subsequent leaf nodes in the tree $\mathcal{S}'$ of $L_{0k}^{'}$ and $L_{kt}^{'}$.\par
There is a relationship between the two trees, and that is:
\begin{align*}
    n_{kt}^{'} &= n_{tk}, \ \ \ 
    L_{kt}^{'} = L_{tk},\ \text{and so,} \\
    \alpha_{kt}^{'} &= \alpha_{tk}
\end{align*}
This is because both subtrees contain the same subleaves, and thus edges entering them. This means that:
\begin{align*}
    n_{0k}^{'} &= \sum_{t=1}^{r_1}n_{kt}^{'} = \sum_{t=1}^{r_1}n_{tk}\\
    &= n_{\bigcdot k}
\end{align*}
and similarly: 
\begin{align*}
    \alpha_{0k}^{'} &= \sum_{t=1}^{r_1}\alpha_{kt}^{'} = \sum_{t=1}^{r_1}\alpha_{tk}\\
    &= \frac{\alpha}{L}\sum_{t=1}^{r_1}L_{tk} = \frac{\alpha}{L}L_{\bigcdot k} = \alpha_{\bigcdot k}
\end{align*}
Now, one can see that:
\begin{align*}
    \overline{\alpha}_0^{'} &= \sum_{k=1}^{r_2} \alpha_{0k}^{'} = \sum_{k=1}^{r_2} \sum_{t=1}^{r_1}\alpha_{kt}^{'} \\
    &= \sum_{k=1}^{r_2} \sum_{t=1}^{r_1}\alpha_{tk} = \sum_{t=1}^{r_1} \sum_{k=1}^{r_2} \alpha_{tk}\\ &= \overline{\alpha}_0
\end{align*}
by switching the indices. This is also true for $\overline{n}_0^{'} = \overline{n}_0$. 

Each child of the root of $\mathcal{S'}_{u}$ is in the same stage $u'$. Summing over the $r_2$ situations, we can calculate the hyperparameters and counts as follows:
\begin{align*}
    \alpha_{ut}^{'} &= \sum_{k=1}^{r_2}\alpha_{kt}^{'} = \sum_{k=1}^{r_2}\alpha_{tk}\\
    &= \alpha_{0t}
\end{align*}
and $n_{ut}^{'} = n_{0t}$, by \eqref{eqn:reverse_ind}. Once again, $\overline{\alpha}_u^{'} = \overline{\alpha}_0^{'} = \overline{\alpha}_0$ and $\overline{n}_u^{'} = \overline{n}_0^{'} = \overline{n}_0$. Hence, the BDepu metric for $\mathcal{S'}_u$ is:
\begin{align*}
    \text{BDepu}(\mathcal{S'}_u) &= \frac{\Gamma(\overline{\alpha}_0^{'})}{\Gamma(\overline{\alpha}_0^{'})+\Gamma(\overline{n}_0^{'})}\left[ \prod_{k=1}^{r_2}\frac{\Gamma(\alpha_{0k}^{'} + n_{0k}^{'})}{\Gamma(\alpha_{0k}^{'})}\right]\frac{\Gamma(\overline{\alpha}_u^{'})}{\Gamma(\overline{\alpha}_u^{'})+\Gamma(\overline{n}_u^{'})}\left[\prod_{t=1}^{r_1}\frac{\Gamma(\alpha_{ut}^{'} + n_{ut}^{'})}{\Gamma(\alpha_{ut}^{'})}\right]\\
    &=\left(\frac{\Gamma(\overline{\alpha}_0)}{\Gamma(\overline{\alpha}_0)+\Gamma(\overline{n}_0)}\right)^2\left[ \prod_{k=1}^{r_2}\frac{\Gamma(\alpha_{\bigcdot k} + n_{\bigcdot k})}{\Gamma(\alpha_{\bigcdot k})}\right]\left[\prod_{t=1}^{r_1}\frac{\Gamma(\alpha_{0 t} + n_{0 t})}{\Gamma(\alpha_{0 t})}\right]\\
    &= \left(\frac{\Gamma(\overline{\alpha}_0)}{\Gamma(\overline{\alpha}_0)+\Gamma(\overline{n}_0)}\right)^2\left[\prod_{t=1}^{r_1}\frac{\Gamma(\alpha_{0t} + n_{0t})}{\Gamma(\alpha_{0t})}\right]\left[ \prod_{k=1}^{r_2}\frac{\Gamma(\alpha_{\bigcdot k} + n_{\bigcdot k})}{\Gamma(\alpha_{\bigcdot k})}\right]
\end{align*}
This is the same as in \eqref{eqn:bd_swap_t}, up to switching $k$ and $t$ in the labels. Therefore BDepu$(\mathcal{S}_u)$=BDepu$(\mathcal{S'}_u)$. \par
\end{proof}
\begin{lemma}
Let $\tau: (\mathcal{S},\theta_{\mathcal{S}}) \rightarrow (\mathcal{S}',\theta_{\mathcal{S}'})$ be a swap operator between staged trees. Then, BDepu($\mathcal{S}$) = BDepu($\mathcal{S}'$).\label{lemma2}
\end{lemma}
\begin{proof}
Assume once again, without loss of generality, that $\tau$ acts on only one twin $\mathcal{S}_u$ in $\mathcal{S}$.

In order to prove $\text{BDepu}(\mathcal{S}) = \text{BDepu}(\mathcal{S'})$, we need to consider the impact of the swap on all situations in $\mathcal{S}$. There of 4 types of situations, those that:
\begin{enumerate}[noitemsep]
    \item Don't lie on any root-to-leaf paths that pass through $\mathcal{S}_u$
    \item Are upstream of $\mathcal{S}_u$ and lie on some root-to-leaf path through $\mathcal{S}_u$
    \item Are in $\mathcal{S}_u$
    \item Are downstream of $\mathcal{S}_u$ and lie on some root-to-leaf path through $\mathcal{S}_u$.
\end{enumerate}
Situations of the first two kinds are unaffected by the swap, and their contribution to the BDepu is unchanged, by independence. In Lemma \ref{lemma1}, we have shown the equivalence between the BDepu contributions for situations in $\mathcal{S}_u$. The only case to show is for those situations that lie on a root-to-leaf path that passed through $\mathcal{S}_u$.\par
Let $\mathcal{L}_u$ be the set of situations in $\mathcal{S}$ that are subleaves $\mathcal{S}_u$. Each of these situations has an edge entering it with some associated count, and that edge (and the situation by extension) represents a unique unfolding of the process to that point. Each subsequent edge for all root-to-leaf paths originating from that situation represents a further evolution of the process. By parameter independence, these edges are independent of the edges upstream. \par
In the swapped subtree $\mathcal{S'}$, these situations and edges still coincide with the same unfolding of the process, and lie at the same distance from the root as before, but may not be in the same location in the tree. Thus, it is akin to reordering the situations at each level that lie on some root-to-leaf path passing through $\mathcal{S}_u$, and equating the florets between the two trees. By the commutativity of multiplication and parameter modularity, the BDepu metric is invariant to such a reordering, and so BDepu($\mathcal{S}$) = BDepu($\mathcal{S'}$). \end{proof}

Proving score equivalence when the swaps operate on a single twin is sufficient because a swap between two trees is simply a composition of such operations.

\begin{lemma}
Let $\kappa: (\mathcal{S},\theta_{\mathcal{S}}) \rightarrow (\mathcal{S}',\theta_{\mathcal{S}'})$ be a resize operator between staged trees that resizes the saturated subtree $(\mathcal{T},\bm{\theta_\mathcal{T}})$ into the floret $(\mathcal{F},\bm{\theta_\mathcal{F}})$. Then, BDepu($\mathcal{T}$) = BDepu($\mathcal{F}$).\label{lemma3}
\end{lemma}
\begin{proof} 
Without loss of generality we assume that $\mathcal{T}$ has at most $2$ levels, as resizes on subtrees with longer paths can be obtained inductively from resizes of subtrees with at most $2$ levels.

Each path $\lambda_i \in \bm{\Lambda}(\mathcal{T})$ corresponds to an edge $e_i$ in $\mathcal{F}$, where the transition probability $\theta(e_i) = \prod_{e \in E(\lambda_i)}\theta(e)$, the probability of the path. 
 Let the number of subleaf nodes in $\mathcal{T}$ be $L_{\mathcal{T}} = \vert \bm{\Lambda}(\mathcal{T})\vert$.
Let there be $r_0$ edges emanating from $s_0$, the root of $\mathcal{T}$. Let each edge $e_{0k}$ have count $n_{0k}$ and hyperparameter $\alpha_{0k}, k = 1,\dots,r_0$. Let $J$ be the number of non-subleaf nodes, the situations, emanating from $s_0$ in $\mathcal{T}$, and so necessarily there are $r_0-J$ paths of length 1 in $\mathcal{T}$. As the subtree is saturated, there are thus $J$ stages at distance 1 from the root, one for each situation. Let each of these singleton stages $u_j = s_j, \ j=1,\dots,J$, have a number of emanating edges $r_j$, with counts $n_{jt}$ and hyperparameters $\alpha_{jt}, t =1,\dots,r_j$. As the tree has a maximum path length of $2$, there are thus $L_{\mathcal{T}} - (r_0-J)$ paths of length 2. This generalisation accounts for the presence of non-stratified trees.\par
As before, the BDepu is calculated stage wise, and so is:
\begin{align}
    \text{BDepu}(\mathcal{T}) = \prod_{j=0}^{J}\left[\frac{\Gamma(\overline{\alpha}_{j})}{\Gamma(\overline{\alpha}_{j}+\overline{n}_{j})} \prod_{k=1}^{r_{j}}\frac{\Gamma(\alpha_{jk}+n_{jk})}{\Gamma(\alpha_{jk})}\right]
    \label{eqn:bd_score_sat}
\end{align}\par
For the paths of length one, which are the edges connecting the root directly to a subleaf, the edge hyperparameters $\alpha_{0k}$, for some $k \in \{1,\dots, r_0\}$, directly contribute to the BDepu.\par
Now, for each $j \in \{1,\dots,J\}$, the non-root singleton stage $u_j$ has some corresponding $i \in \{1,\dots,r_0\}$, the edge entering it, such that $\alpha_{0i} = \overline{\alpha}_j$ and $n_{0i} = \overline{n}_j$. Thus, these terms cancel in the BDepu. Hence, only the outgoing edge hyperparameters and counts for these singleton stages, $\alpha_{jk}$ and $n_{jk}$, contribute to the BDepu. These edges must necessarily enter subleaves. \par
 As such, the edges entering the subleaves are the only edges which contribute to the BDepu. We can label each of these edges as $e_{k}^{'}, \ k=1,\dots,L_\mathcal{T}$, with count $n_{k}^{'}$ and hyperparameter $\alpha_{k}^{'}$. The BDepu for $\mathcal{T}$ is thus:
\begin{align}
    \text{BDepu}(\mathcal{T}) =  \frac{\Gamma(\overline{\alpha}_{0})}{\Gamma(\overline{\alpha}_{0}+\overline{n}_{0})} \prod_{k=1}^{L_\mathcal{T}} \frac{\Gamma(\alpha_{k}^{'}+n_{k}^{'})}{\Gamma(\alpha_{k}^{'})}.
    \label{eqn:bd_score_sat2}
\end{align}
The resized floret $\mathcal{F}$ still has a root $s_{0}$, with $L_\mathcal{T}$ outgoing edges, one for each path. Each of these edges necessarily enters a subleaf, and so has the same count and hyperparameter as the edge entering the subleaf in $\mathcal{T}$. Thus, when the BDepu is calculated, it is equivalent to the expression in \eqref{eqn:bd_score_sat2}, and $\text{BDepu}(\mathcal{T}) = \text{BDepu}(\mathcal{F})$, as required. \par
\end{proof}

The cancellation property between the hyperparameter of the edge entering the singleton stage and the hyperparameter for the singleton stage itself is indicative of the desire for mass conservation. If this property is not in place, these two terms will not necessarily cancel and score equivalence will not hold. Furthermore, the fact the BDepu only accounts for the edges entering the subleaves is why the formulation of the hyperparameters for the CS-BDeu in Section \ref{subsec:realted_work} was also not score-equivalent despite the mass conservation property being in place.

\begin{lemma}
Let $\kappa: (\mathcal{S},\theta_{\mathcal{S}}) \rightarrow (\mathcal{S}',\theta_{\mathcal{S}'})$ be a resize operator between staged trees that resizes the saturated subtree $(\mathcal{T},\bm{\theta_\mathcal{T}})$ into the floret $(\mathcal{F},\bm{\theta_\mathcal{F}})$. Then, BDepu($\mathcal{S}$) = BDepu($\mathcal{S}'$).\label{lemma4}
\end{lemma}
\begin{proof}
In Lemma \ref{lemma3}, we have proven the BDepu contributions for the saturated subtree and the resized floret are equivalent. The argument then follows similarly to Lemma \ref{lemma2}, essentially that due to the fact the BDepu, is calculated stage-wise, the rest of the BDepu is unaffected by the resize. Firstly, as $\mathcal{T}$ is saturated, the stages within it are not repeated throughout the tree. For stages that are upstream of $\mathcal{T}$ or do not lie on any paths that pass through $\mathcal{T}$, their contributions to the BDepu are unaffected. As the prior hyperparameters and counts at the subleaves of the subtrees $\mathcal{T}$ and the florets $\mathcal{F}$ will be the same, the BDepu contibutions for any downstream stages will be the same. Hence, the BDepu is score-equivalent under the resize operator, BDepu($\mathcal{S}$) = BDepu($\mathcal{S}'$).
\end{proof}
We thus have proven the resize operator, when acting on a saturated subtree, is score-equivalent. Now, we suppose the more general form where it operates on several conditionally saturated subtrees.
\begin{lemma}
Let $u$ be a potentially singleton stage in $(\mathcal{S},\theta_{\mathcal{S}})$ with index set $Q_u$ detailing the situations that are in $u$. Suppose each $s_q, \ i \in Q_u$ has a conditionally saturated subtree $(\mathcal{T}_q,\bm{\theta_{\mathcal{T}_q}})$ emanating from it. Let $\kappa: (\mathcal{S},\theta_{\mathcal{S}}) \rightarrow (\mathcal{S}',\theta_{\mathcal{S}'})$ be a resize operator between staged trees that resizes each conditionally saturated subtree $(\mathcal{T}_q,\bm{\theta_{\mathcal{T}_q}})$ into the floret $(\mathcal{F}_q,\bm{\theta_{\mathcal{F}_q}}), \  \forall q \in Q_u$. Then, BDepu($\mathcal{S}$) = BDepu($\mathcal{S}'$).\label{lemma5}
\end{lemma}
\begin{proof}
Each of the subtrees $\mathcal{T}_q$ must necessarily have the same topology and prior hyperparameters on the edges and situations, but with possibly different counts. As the subtrees are conditionally saturated, no stage is repeated inside each $\mathcal{T}_q$, and no stage in $\mathcal{T}_q$ is repeated throughout the rest of $\mathcal{S}$. \par
As the BDepu is calculated stage wise, there is not a BDepu contribution for each $\mathcal{T}_q$, but one contribution which combines them all, as if the collection were a single subtree. We will call this subtree $\mathcal{T}_u$. The hyperparameters and counts for each edge in $\mathcal{T}_u$ will be the sum of each of the corresponding hyperparameters and counts across all trees. Now, as $\mathcal{T}_q$ are conditionally saturated, $\mathcal{T}_u$ is saturated.  \par
 We can also resize each $\mathcal{T}_q$ into a floret $\mathcal{F}_q$, with each root in the same stage. Hence, we combine the florets into one, $\mathcal{F}_u$.\par
 In terms of the BDepu calculation, both $\mathcal{T}_u$ and $\mathcal{F}_u$ behave in the same way as a saturated subtree. Hence, they are subject to the same cancellation properties laid out in Lemma \ref{lemma3}, and we get that BDepu($\mathcal{T}_u$) = BDepu($\mathcal{F}_u$). We then use Lemma \ref{lemma4} to get that BDepu($\mathcal{S}$) = BDepu($\mathcal{S}'$).
\end{proof}

\begin{theorem}
 Suppose $(\mathcal{S},\theta_{\mathcal{S}})$ and $(\mathcal{S}',\theta_{\mathcal{S}'})$ are statistically equivalent staged trees. Then, BDepu($\mathcal{S}$) = BDepu($\mathcal{S}'$).\label{thm1}
 \end{theorem}
 \begin{proof} By Lemmas \ref{lemma2} and \ref{lemma5}, the BDepu metric is invariant to swaps and resizes. In \citet{gorgen18}, it is proven that two staged trees are statistically equivalent if and only if the map between them is a finite composition of swaps and resizes. Hence, the BDepu metric is score-equivalent, and BDepu($\mathcal{S}$) = BDepu($\mathcal{S'}$). 
\end{proof}

Furthermore, following from our proof that the BDepu and CS-BDeu are equivalent on stratified staged trees in \ref{sec:app:csbdeu}, we claim the CS-BDeu is score-equivalent for two staged trees belonging to the same statistical equivalence class if and only if they are both stratified. 

\section{Analysing a Real World Dataset}
The \emph{Titanic} data set contains information on the fates of passengers on the ocean liner Titanic, summarised according to economic status, age, sex, and whether they survived or not. This data set has previously been analysed by \citet{stagedtrees20} for staged trees. We will use it to demonstrate the swap and resize operators, and how the BDepu is score-equivalent under both.\par
There are 3 covariates chosen, alongside the response variable of whether an individual survived or not. We call this response Survival, with values of Yes or No. The three covariates, in the order they appear in the tree, are:
\begin{itemize}
\itemsep0em
    \item Role: Crew or Passenger
    \item Sex: Male or Female
    \item Age: Adult or Child
\end{itemize}
Notably, there are no children in the crew, and so this brings structural zeroes into the data set. We will represent this by deleting the edge associated to Child in the event tree for those who are Crew. This means the event tree is not $\mathcal{X}$-compatible, where $\mathcal{X} = ($Role, Sex, Age, Survival$)$, and so the staged tree is non-stratified. \par
We then embellish the event tree with colours to create a staged tree, where two situations of the same colour are in the same stage. In order to decide on this stage structure, we propose two conditional independence statements that will dictate the resultant staging. The staged tree after these conditional independence statements are asserted, $\mathcal{S}_1$, can be seen in Figure \ref{fig:tree_1}, with edge counts in parentheses.\par
First, we will assume Age and Sex are independent, given the Role is known. That is, $\text{Age} \ \indep \ \text{Sex} \ \vert \ \text{Role}$. For Crew members, this is trivial, as they must necessarily be Adults, but for Passengers this creates a stage for $s_5$ and $s_6$.\par
The second conditional independence statement will be based on the fact that women and children were given access to the lifeboats. That is, children of both sexes were given equal access, which was not the case for adults. As such, we have the context-specific conditional independence statement that $\text{Survival} \ \indep \ \text{Sex} \ \vert \ \text{Role}, \ \text{Age = Child}$. This means $s_{10}$ and $s_{12}$ are in the same stage. Note that the inclusion of role here is superfluous, as all children must be Passengers. \par
\label{sec:example}
\begin{figure}[!ht]
   \centering
\scalebox{0.4}{  
\xymatrixcolsep{11pc} \xymatrix{
\textbf{\huge Role}&\textbf{\huge Sex}&\textbf{\huge Age}&\textbf{\huge Survival}\\
 &&\text{\huge \color{red} $s_{3}$} \ar[r]|-{\txt{\huge Adult (862)}}
 &\text{\huge \color{Gray} $s_{7}$}  \ar[dr]|-{\txt{\huge  No (670)}}
			        \ar[r]|-{\txt{\huge Yes (192)}}
 &\text{\huge $l_{1}$}\\
 &&&&\text{\huge $l_{2}$}\\
 &\text{\huge \color{orange} $s_{1}$} \ar[uur]|-{\txt{\huge   Male (862)}}
			       \ar[ddr]|-{\txt{\huge Female (23)}}&&&\\
 &&&&\text{\huge $l_{3}$}\\
 &&\text{\huge \color{red} $s_{4}$} \ar[r]|-{\txt{\huge Adult (23)}}
 &\text{\huge \color{Cyan} $s_{8}$}  \ar[r]|-{\txt{\huge  No (3)}}
			        \ar[ur]|-{\txt{\huge  Yes (20)}}
 &\text{\huge $l_{4}$}\\
\text{\huge $s_{0}$} \ar[uuur]|-{\txt{\huge Crew (885)}}
			        \ar[dddr]|-{\txt{\huge Passenger (1316)}}&&&&\\
 &&&&\text{\huge $l_{5}$}\\
 &&&\text{\huge \color{Goldenrod} $s_{9}$} \ar[ur]|-{\txt{\huge  Yes (146)}}
				\ar[r]|-{\txt{\huge No (659)}}
 &\text{\huge $l_{6}$}\\
&\text{\huge \color{ForestGreen} $s_{2}$} \ar[r]|-{\txt{\huge Male (869)}}
			       \ar[dddr]|-{\txt{\huge  Female (447)}}
&\text{\huge \color{blue} $s_{5}$} \ar[ur]|-{\txt{\huge  Adult (805)}}
			        \ar[r]|-{\txt{\huge  Child (64)}}
 &\text{\huge \color{CarnationPink} $s_{10}$}  \ar[r]|-{\txt{\huge  Yes (29)}}
			        \ar[dr]|-{\txt{\huge No (35)}}
 &\text{\huge $l_{7}$}\\
 &&&&\text{\huge $l_{8}$}\\
 &&&&\text{\huge $l_{9}$}\\
 &&\text{\huge \color{blue} $s_{6}$} \ar[dr]|-{\txt{\huge Child (45)}}
			        \ar[r]|-{\txt{\huge Adult (402)}}
 &\text{\huge \color{LimeGreen} $s_{11}$}  \ar[r]|-{\txt{\huge  No (106)}}
			        \ar[ur]|-{\txt{\huge Yes (296)}}
 &\text{\huge $l_{10}$}\\
 &&&\text{\huge \color{CarnationPink} $s_{12}$} \ar[r]|-{\txt{\huge Yes (28)}}
				\ar[dr]|-{\txt{\huge No (17)}}
 &\text{\huge $l_{11}$}\\
 &&&&\text{\huge $l_{12}$}\\
} }
\caption{\label{fig:tree_1} The staged tree $\mathcal{S}_1$ for the Titanic data set with edge counts in parentheses.}
\end{figure}
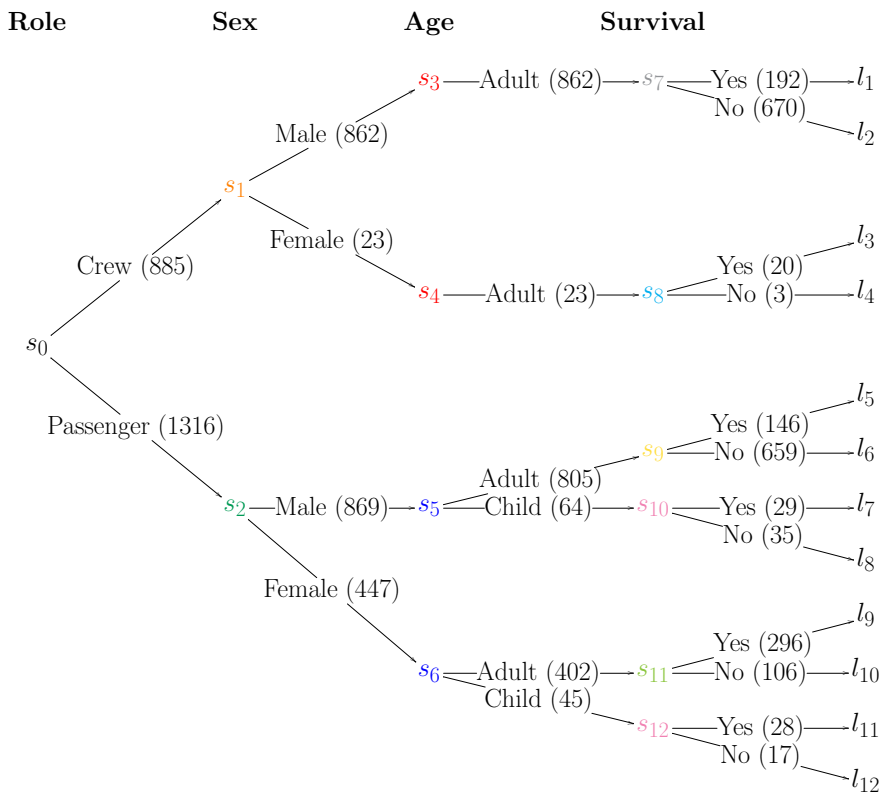
Now, we can see that there are twins around stages $u_{4} = \{s_3, s_4\}$ and $u_{5} = \{s_5, s_6\}$. The swap operator can be applied to both of these twins, which means all situations associated with Age and Sex will be swapped, a level swap \citep{gorgen18}. This corresponds to a reordering of the variables in the tree, and thus a different causal interpretation of the process. We perform these two swap operators to create a statistically equivalent tree, $\mathcal{S}_2$, which can be seen in Figure \ref{fig:tree_2}. \par
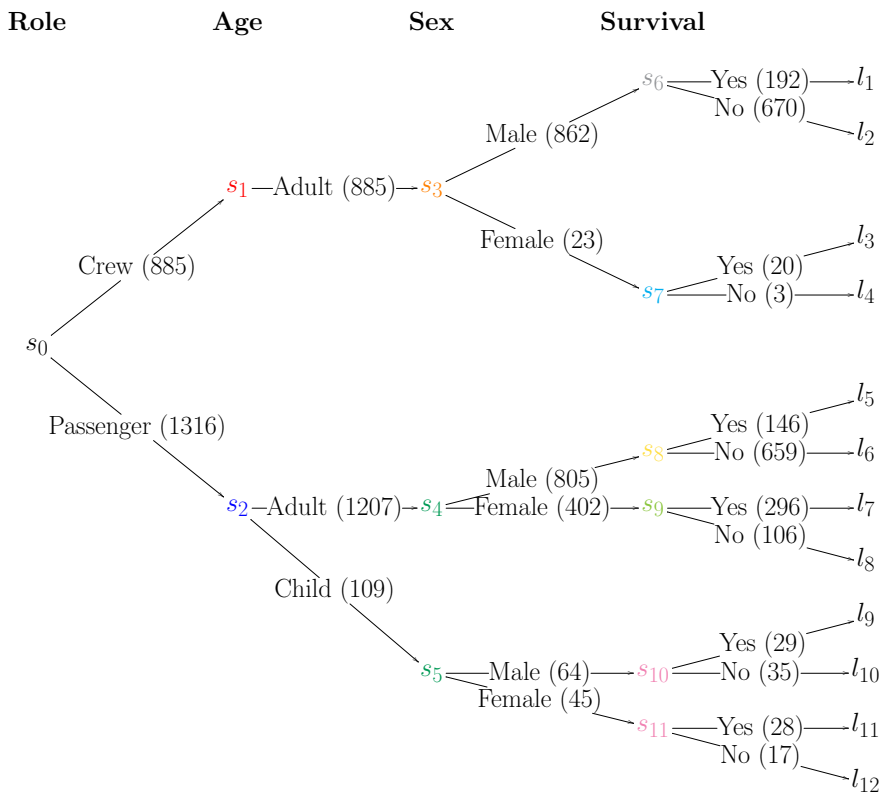
\begin{figure}[!ht]
   \centering
\scalebox{0.4}{  
\xymatrixcolsep{11pc} \xymatrix{
\textbf{\huge Role}&\textbf{\huge Age}&\textbf{\huge Sex}&\textbf{\huge Survival}\\
 &&&\text{\huge \color{Gray} $s_{6}$}  \ar[dr]|-{\txt{\huge No (670)}}
			        \ar[r]|-{\txt{\huge Yes (192)}}
 &\text{\huge $l_{1}$}\\
 &&&&\text{\huge $l_{2}$}\\
 &\text{\huge \color{red} $s_{1}$} \ar[r]|-{\txt{\huge  Adult (885)}}&
 \text{\huge \color{orange} $s_{3}$} \ar[uur]|-{\txt{\huge  Male (862)}}
 \ar[ddr]|-{\txt{\huge Female (23)}}&&\\
 &&&&\text{\huge $l_{3}$}\\
 &&&\text{\huge \color{Cyan} $s_{7}$}  \ar[r]|-{\txt{\huge   No (3)}}
			        \ar[ur]|-{\txt{\huge Yes (20)}}
 &\text{\huge $l_{4}$}\\
\text{\huge $s_{0}$} \ar[uuur]|-{\txt{\huge Crew (885)}}
			        \ar[dddr]|-{\txt{\huge Passenger (1316)}}&&&&\\
 &&&&\text{\huge $l_{5}$}\\
 &&&\text{\huge \color{Goldenrod} $s_{8}$} \ar[ur]|-{\txt{\huge  Yes (146)}}
				\ar[r]|-{\txt{\huge No (659)}}
 &\text{\huge $l_{6}$}\\
&\text{\huge \color{blue} $s_{2}$} \ar[r]|-{\txt{\huge  Adult (1207)}}
			       \ar[dddr]|-{\txt{\huge  Child (109)}}
&\text{\huge \color{ForestGreen} $s_{4}$} \ar[ur]|-{\txt{\huge  Male (805)}}
			        \ar[r]|-{\txt{\huge Female (402)}}
 &\text{\huge \color{LimeGreen} $s_{9}$}  \ar[r]|-{\txt{\huge   Yes (296)}}
			        \ar[dr]|-{\txt{\huge  No (106)}}
 &\text{\huge $l_{7}$}\\
 &&&&\text{\huge $l_{8}$}\\
 &&&&\text{\huge $l_{9}$}\\
 &&\text{\huge \color{ForestGreen} $s_{5}$} \ar[dr]|-{\txt{\huge Female (45)}}
			        \ar[r]|-{\txt{\huge Male (64)}}
 &\text{\huge \color{CarnationPink} $s_{10}$}  \ar[r]|-{\txt{\huge   No (35)}}
			        \ar[ur]|-{\txt{\huge   Yes (29)}}
 &\text{\huge $l_{10}$}\\
 &&&\text{\huge \color{CarnationPink} $s_{11}$} \ar[r]|-{\txt{\huge  Yes (28)}}
				\ar[dr]|-{\txt{\huge  No (17)}}
 &\text{\huge $l_{11}$}\\
 &&&&\text{\huge $l_{12}$}\\
} }
\caption{\label{fig:tree_2} The staged tree $\mathcal{S}_2$ for the Titanic data set after performing a swap on $\mathcal{S}_1$.}
\end{figure}
In $\mathcal{S}_2$, we can see that the subtree rooted at $s_1$ is saturated, and so a resize can be applied. As there are four paths in this subtree, the resultant floret will have four outgoing edges, with each edge representing a unique combination of Sex and Survival. The staged tree $\mathcal{S}_3$ created after performing this resize on $\mathcal{S}_2$ can be seen in Figure \ref{fig:tree_3}. We note that now each level of the tree can no longer be associated with a specific covariate, even after accounting for the structural zeroes. We also point out that the edge counts for the resized floret rooted at $s_1$ match the edge counts entering the leaves $l_1 - l_4$ in $\mathcal{S}_2$. This exemplifies the cancellation property seen in the proof of score equivalence for the resize operator in Lemma \ref{lemma3}, where the only edge counts that are important are the ones entering the leaves. \par
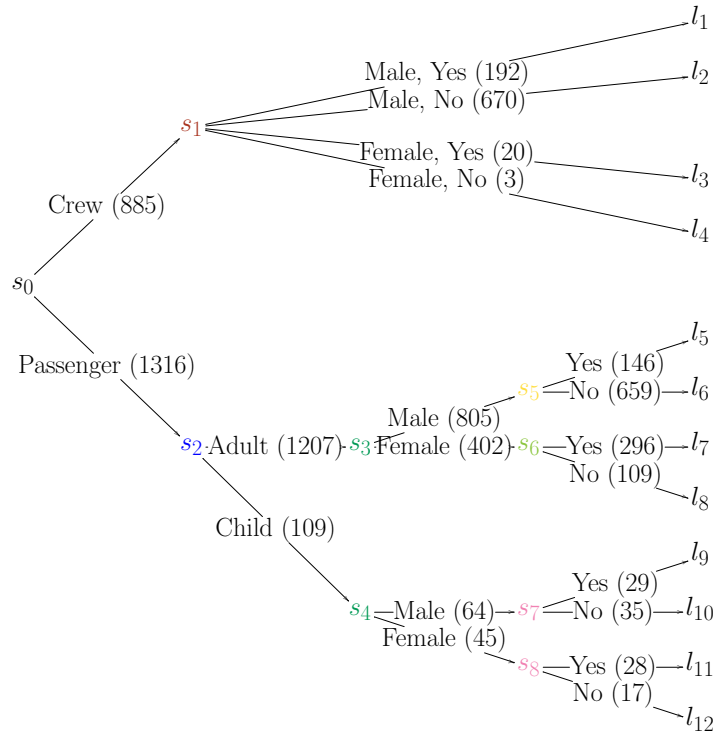
\begin{figure}[!ht]
   \centering
\scalebox{0.4}{  
\xymatrixcolsep{11pc} \xymatrix{
 &&&&\text{\huge $l_{1}$}\\
 &&&&\text{\huge $l_{2}$}\\
 &\text{\huge \color{Mahogany} $s_{1}$} \ar[uurrr]|-{\txt{\huge    Male, Yes (192)}}
 \ar[urrr]|-{\txt{\huge   Male, No (670)}}
 \ar[drrr]|-{\txt{\huge   Female, Yes (20)}}
 \ar[ddrrr]|-{\txt{\huge   Female, No (3)}}&&&\\
 &&&&\text{\huge $l_{3}$}\\
 &&&&\text{\huge $l_{4}$}\\
\text{\huge $s_{0}$} \ar[uuur]|-{\txt{\huge Crew (885)}}
			        \ar[dddr]|-{\txt{\huge Passenger (1316)}}&&&&\\
 &&&&\text{\huge $l_{5}$}\\
 &&&\text{\huge \color{Goldenrod} $s_{5}$} \ar[ur]|-{\txt{\huge   Yes (146)}}
				\ar[r]|-{\txt{\huge  No (659)}}
 &\text{\huge $l_{6}$}\\
&\text{\huge \color{blue} $s_{2}$} \ar[r]|-{\txt{\huge Adult (1207)}}
			       \ar[dddr]|-{\txt{\huge  Child (109)}}
&\text{\huge \color{ForestGreen} $s_{3}$} \ar[ur]|-{\txt{\huge   Male (805)}}
			        \ar[r]|-{\txt{\huge Female (402)}}
 &\text{\huge \color{LimeGreen} $s_{6}$}  \ar[r]|-{\txt{\huge   Yes (296)}}
			        \ar[dr]|-{\txt{\huge  No (109)}}
 &\text{\huge $l_{7}$}\\
 &&&&\text{\huge $l_{8}$}\\
 &&&&\text{\huge $l_{9}$}\\
 &&\text{\huge \color{ForestGreen} $s_{4}$} \ar[dr]|-{\txt{\huge  Female (45)}}
			        \ar[r]|-{\txt{\huge Male (64)}}
 &\text{\huge \color{CarnationPink} $s_{7}$}  \ar[r]|-{\txt{\huge No (35)}}
			        \ar[ur]|-{\txt{\huge Yes (29)}}
 &\text{\huge $l_{10}$}\\
 &&&\text{\huge \color{CarnationPink} $s_{8}$} \ar[r]|-{\txt{\huge Yes (28)}}
				\ar[dr]|-{\txt{\huge No (17)}}
 &\text{\huge $l_{11}$}\\
 &&&&\text{\huge $l_{12}$}\\
} }
\caption{\label{fig:tree_3} The staged tree $\mathcal{S}_3$ for the Titanic data set after performing a resize on $\mathcal{S}_2$.}
\end{figure}
We will now calculate the BDepu for each model. The prior hyperparameters are calculated according to \eqref{eqn:BDepu_setting}, which requires the imaginary sample size $\alpha$ to be specified. 
As this is purely a demonstration of score equivalence, we will choose $\alpha = 12$, the number of leaves, so as to avoid fractions in the calculation. In Section \ref{sec:discussion}, we will discuss choices of $\alpha$.\par
Above, we showed that for $u_1 = \{s_0\}$, we have that $\alpha_{11} = 4$ and $\alpha_{12} = 8$. For $u_4 = \{s_3, s_4\}$, we have that there is only one hyperparameter, $\alpha_{41} = 4$. For $u_5 = \{s_5, s_6\}$, we have two hyperparameters, with $\alpha_{51} = \alpha_{52} = 4$. We have placed all three BDepu calculations in \ref{sec:app:calcs}. \par
When we inspect BDepu($\mathcal{S}_1$) and BDepu($\mathcal{S}_2$), \eqref{eqn:bd_s_1} and \eqref{eqn:bd_s_2} respectively, we can see that they are the same, except with some terms rearranged. This highlights the idea that a swap simply rearranges the terms in the BDepu, without materially affecting the calculation.\par
With regard to the BDepu for $\mathcal{S}_3$, the contributions from the parts of the tree that have not been resized are the same as in $\mathcal{S}_2$. The focus is then on the subtree that was resized, and its contribution in both $\mathcal{S}_2$ and $\mathcal{S}_3$. We can see from \eqref{eqn:bd_s_3} that the contribution for the resized subtree and the subsequent floret are the same, after cancellation. Hence, we have that BDepu($\mathcal{S}_1$) = BDepu($\mathcal{S}_2$)=BDepu($\mathcal{S}_3$).

\section{Discussion}
\label{sec:discussion}
In this paper we have presented the BDepu, based on the BD-metric, which is the first score-equivalent scoring function defined for staged trees. The BDepu for staged trees is analogous to the score-equivalent BDeu for Bayesian networks. Nonetheless, the methods in this paper were all centred on a default prior specification, and there are still avenues for research into the impacts of how this prior is specified, and implications of specifying more informative priors. 

In order to set this default prior, an imaginary sample size $\alpha$ must be chosen. Prominent methods for setting $\alpha$ include those of \citet{neap03} for BNs, who claimed that the maximum number of categories for a variable is the theoretical minimum number of observations needed in a data set, and thus the minimum confidence that should be given to the prior. However, we believe this approach has issues when applied to staged trees, non-stratified in particular, even when the number of outgoing edges from each situation is considered analogous to the number of categories for an associated variable. In fact, we claim that any default way of setting this imaginary sample size can have undesirable properties. The MAP structure chosen can be highly sensitive to the choice of $\alpha$, and this is shown for BNs in \citet{silander07sensi}. Hence, a sensitivity analysis should always be performed to investigate the effects of the choice of $\alpha$, rather than settling for one heuristic.

The BDepu metric is a specific form of the BD-metric, and is one instance where the BD-metric is score equivalent for staged trees. It is also of interest to explore the general properties of score-equivalent metrics for staged trees. Notably, the BDeu metric for BNs is a special case of the the BDe metric \citep{heckerman95}, where the BDe metric apportions $\alpha$ according to some pre-specified vector of prior probabilities, rather than uniformly. For BNs, the BDe metric has also been proven to be score-equivalent \citep{chickering95}, and has an analogous form for staged trees, similar to the BDepu, where $\alpha$ is instead distributed across each path based on a prior probability of traversing that path. However, just as the BDeu is more commonly used for BNs than the BDe, we expect that the BDepu would be more applicable generally than such an informative approach, as such prior probabilities may not be readily available. If they were available, this would often be in cooperation with a domain expert, who may feel more comfortable placing informative priors on the stages rather than the paths, as it is more natural to consider the conditional probabilities. These individual stage priors can also have their own measures of confidence.

If informative priors on the stages are asserted, as the BD-metric is calculated stage wise, it would appear sensible to ignore the paths and just use the BD-metric itself. Notably, a key point of the proof of score equivalence provided in this paper is centred on mass conservation. Hence, we conjecture the metric will not be score-equivalent in general unless the prior conserves mass. As such, informative priors would need to respect the causal ordering of the tree in order to be mass conserving and thus score-equivalent. Therefore, setting informative priors that are not mass conserving could be seen as containing causal information, with respect to the BD-metric. Recalling the fact that mass conservation was needed only for the resize operator, further thought should be given to the meaning of an informative prior after resizing.   



\appendix

\section{CS-BDeu Alternative Specifications}\label{sec:app:alt_spec}

We previously outlined in Section \ref{sec:app:csbdeu} one possible modification to the hyperparameters in the CS-BDeu that could be applied to all staged trees. There are two other options that could also be arrived at, but neither is score-equivalent for non-stratified staged trees, as we will show with some counter-examples.

For the first option, at each level, we distribute the imaginary sample size $\alpha$ amongst each of the situations, and then at each situation distribute this weight across its outgoing edges. We focus again on the staged trees in Figure \ref{fig:cs_example}, and note that this method of setting the hyperparameters is equivalent to the specification laid out previously. As shown in that example, the CS-BDeu is not score-equivalent, and so this method of setting the hyperparameters will not be score-equivalent either. We also point out that this method conserves mass in this example, but this will not be true in general. 

The second option is to instead, at each level, distribute $\alpha$ amongst the outgoing edges, and then the hyperparameter for a situation will be the sum of its outgoing edges. We revisit again Example \ref{ex:csbdeu_counter_example}, where this way of setting the hyperparameters can be seen in Figure \ref{fig:cs_example2}. As can be seen, the hyperparameters entering the leaves between both trees is now the same, unlike with the two other approaches. However, because $\alpha$ has not been backpropagated through the whole tree, only from edge to its origin situation, the mass conservation property does not hold. In \eqref{eqn:cs_calc}, we see that CS-BDeu$(\mathcal{S},\mathcal{D};\bm{\alpha}) \neq $CS-BDeu$(\mathcal{S}',\mathcal{D};\bm{\alpha})$. 

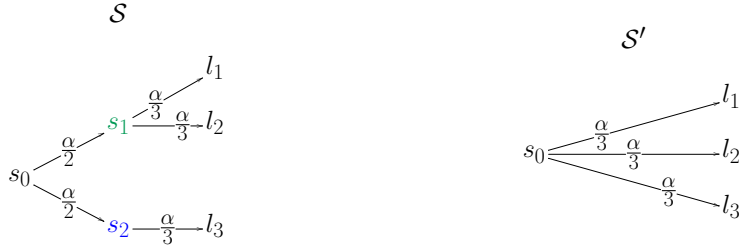
\begin{figure}[!ht]
\centering
\begin{minipage}{0.5 \textwidth}
\centering
\scalebox{0.4}{  
\xymatrixcolsep{5.5pc} \xymatrix{
&\text{\huge $\mathcal{S}$}&\\
 &&\text{\huge $l_1$}\\
 &\text{\huge  \color{ForestGreen} $s_{1}$} \ar[r]|-<(0.7){\txt{\huge $\tfrac{\alpha}{3}$}}
			        \ar[ur]|-<(0.3){\txt{\huge $\tfrac{\alpha}{3}$}}
 &\text{\huge  $l_2$}\\
   \text{\huge  $s_{0}$} \ar[ur]|-{\txt{\huge  $\tfrac{\alpha}{2}$}}
                      \ar[dr]|-{\txt{\huge $\tfrac{\alpha}{2}$}}&&\\
&\text{\huge \color{blue} $s_{2}$}\ar[r]|-{\txt{\huge  $\tfrac{\alpha}{3}$}} &\text{\huge $l_3$}\\
} } 
\end{minipage}\hfill
\begin{minipage}{0.5 \textwidth}
\centering
\scalebox{0.4}{  
\xymatrixcolsep{5.5pc} \xymatrix{
&\text{\huge $\mathcal{S}'$}&\\
 &&\text{\huge $l_1$}\\
 \text{\huge  $s_{0}$} \ar[urr]|-<(0.3){\txt{\huge  $\tfrac{\alpha}{3}$}}
                        \ar[rr]|-{\txt{\huge  $\tfrac{\alpha}{3}$}}
                      \ar[drr]|-<(0.7){\txt{\huge $\tfrac{\alpha}{3}$}}
 &&\text{\huge  $l_2$}\\
&&\text{\huge $l_3$}\\
} } 
\end{minipage}
\caption{The staged trees $\mathcal{S}$ and $\mathcal{S}'$ labelled with the edge hyperparameters under the CS-BDeu.}
\label{fig:cs_example2}
\end{figure}

\footnotesize
\begin{align}
    \text{CS-BDeu}(\mathcal{S},\mathcal{D};\bm{\alpha}) = &\left[\frac{\Gamma(\alpha)}{\Gamma(\alpha+3)}\frac{\Gamma(\tfrac{\alpha}{2}+2)\Gamma(\tfrac{\alpha}{2}+1)}{\Gamma(\tfrac{\alpha}{2})^{2}}\right]
    \left[\frac{\Gamma(\tfrac{2\alpha}{3})}{\Gamma(\tfrac{2\alpha}{3}+2)}\frac{\Gamma(\tfrac{\alpha}{3}+1)^2}{\Gamma(\tfrac{\alpha}{3})^2}\right]\nonumber\\
    &\left[\frac{\Gamma(\tfrac{\alpha}{3})}{\Gamma(\tfrac{\alpha}{3}+1)}\frac{\Gamma(\tfrac{\alpha}{3}+1)}{\Gamma(\tfrac{\alpha}{3})}\right]\nonumber\\
    =&\frac{\Gamma(\alpha)\Gamma(\tfrac{2\alpha}{3})}{\Gamma(\alpha+3)\Gamma(\tfrac{2\alpha}{3}+2)}\frac{\Gamma(\tfrac{\alpha}{2}+2)\Gamma(\tfrac{\alpha}{2}+1)\Gamma(\tfrac{\alpha}{3}+1)^2}{\Gamma(\tfrac{\alpha}{2})^{2}\Gamma(\tfrac{\alpha}{3})^{2}} \nonumber\\
    \text{CS-BDeu}(\mathcal{S}',\mathcal{D};\bm{\alpha}) = &\frac{\Gamma(\alpha)}{\Gamma(\alpha+3)}\frac{\Gamma(\tfrac{\alpha}{3}+1)^3}{\Gamma(\tfrac{\alpha}{3})^3} \neq \text{CS-BDeu}(\mathcal{S},\mathcal{D};\bm{\alpha}) \label{eqn:cs_calc}
\end{align}
 
\normalsize

\section{CS-BDeu vs BDepu for Stratified Staged Trees}
\label{sec:app:csbdeu}
Assume we have a stratified staged tree $\mathcal{S}$ with $m$ levels, as in Section \ref{sec:BDeu}. As it is stratified, we have that, for the $j$th stage at level i, corresponding to $X_i$, $r_{ij} = r_i$, where $r_i$ is the size of the state space for $X_i$. We set the hyperparameters $\alpha_{ijk}$ for the CS-BDeu as described previously. We will define this recursively. First, we focus on the only stage at level 1, corresponding to the root.
\begin{align*}
    \alpha_{11k} = \tfrac{\alpha}{r_1}, k = 1,\dots,r_1
\end{align*}

Then, for the stage $u_{2j}$, which contains $h_{2j}$ stages, it has $h_{2j}$ incoming edges also., with each edge having a weight of $\tfrac{\alpha}{r_1}$. As such, we have that:
\begin{align*}
    \dot{\alpha}_{2j} = \tfrac{\alpha h_{2j} }{r_1}
\end{align*}
and so
\begin{align*}
    \alpha_{2jk} = \tfrac{\alpha h_{2j} }{r_1r_2}
\end{align*}
This is the stage hyperparameter for the edge, which means each individual edge (emanating from the constituent stages) has weight $\tfrac{\alpha}{r_1r_2}$. This relationship is true for each level. That is, for level $i$, each edge has weight $\tfrac{\alpha}{\prod_{t=1}^{i+1}r_t}$, and so for the stage $u_{ij}$ with $h_{ij}$ constituent situations and thus incoming edges, we have that:
\begin{align}
    \dot{\alpha}_{ij} &= \frac{\alpha h_{ij}}{\prod_{t=1}^{i-1}r_t},\\
    \alpha_{ijk} &= \frac{\dot{\alpha}_{ij}}{r_{i}} = \frac{\alpha h_{ij}}{\prod_{t=1}^{i}r_t}
\end{align}


We propose that this formulation of the hyperparameters is equivalent to the one we proposed in \eqref{eqn:BDepu_setting} for stratified staged trees. \par
As the staged tree is stratified, the number of leaves $L = \prod_{t=1}^{m}r_t$. For a general situation $s$ at level $i$, the number of paths it is a constituent of is $\prod_{t=i}^{m}r_t$. Hence, for a general edge $e$ emanating from $s$, the number of paths it is a constituent of is $\prod_{t=i+1}^{m}r_t$. As such, for a stage $u_j$ where all of its $h_j$ situations, labelled $s_{j}^{m}$ are at level $i$ with $r_i$ outgoing edges, the hyperparameters are:
\begin{align*}
    \alpha_{jk} &= \frac{\alpha h_{j} \prod_{t=i+1}^{m}r_t}{\prod_{t=1}^{m}r_t} =  \frac{\alpha h_{j} }{\prod_{t=1}^{i}r_t}\\
    \overline{\alpha}_j &= \sum_{k=1}^{r_i}\alpha_{jk} = \frac{\alpha h_j}{\prod_{t=1}^{i-1}r_t}
\end{align*}
as required.\par
We can see in the two formulations that, while they are equivalent for stratified staged stages, the difference is the order in which the hyperparameters are calculated. For the CS-BDeu, we propagate $\alpha$ forward from the root, from situation to outgoing edge, while for the BDepu we propagate $\alpha$ backwards from the leaves. For stratified staged trees there are no differences in either approach. 

 Hence, by the score equivalence proved for the BDepu, and the fact that the CS-BDeu and BDepu are equivalent on stratified staged trees, the CS-BDeu is score-equivalent when comparing stratified staged trees.

\section{BDepu Calculations}\label{sec:app:calcs}
The BDepu for $\mathcal{S}_1$ is:
\footnotesize
\begin{align}
    \text{BDepu}(\mathcal{S}_1) = &\left[\frac{\Gamma(12)}{\Gamma(12+2201)}\frac{\Gamma(4+885)\Gamma(8+1316)}{\Gamma(4)\Gamma(8)}\right]\nonumber\\
        &\left[\frac{\Gamma(4)}{\Gamma(4+885)}\frac{\Gamma(2+862)\Gamma(2+23)}{\Gamma(2)^2}\frac{\Gamma({8})}{\Gamma( {8+1316})}\frac{\Gamma( {4+869})\Gamma( {4+447})}{\Gamma( {4})^2}\right]\nonumber\\
        &\left[ \frac{\Gamma( {4})}{\Gamma( {4+885})}\frac{\Gamma( {4+885})}{\Gamma( {4})}\frac{\Gamma( {8})}{\Gamma( {8+1316})}\frac{\Gamma( {4+1207})\Gamma( {4+109})}{\Gamma( {4})^{2}}
       \right]\nonumber\\
        &\left[\frac{\Gamma( {2})}{\Gamma( {2+862})}\frac{\Gamma( {1 +192})\Gamma( {1+670})}{\Gamma( {1})^{2}}
        \frac{\Gamma( {2})}{\Gamma( {2+23)}}\frac{\Gamma( {1+20})\Gamma( {1+3})}{\Gamma( {1})^{2}}
        \right]\nonumber \\
        &\left[\frac{\Gamma( {2})}{\Gamma( {2+805})}\frac{\Gamma( {1+146})\Gamma( {1+659})}{\Gamma( {1})^{2}}
        \frac{\Gamma( {4})}{\Gamma( {4+109)}}\frac{\Gamma( {2+57})\Gamma( {2+52})}{\Gamma( {2})^{2}}
        \right]\nonumber \\
        &\left[\frac{\Gamma( {2})}{\Gamma( {2+402})}\frac{\Gamma( {1+296})\Gamma( {1+106})}{\Gamma( {1})^{2}}
        \right]\label{eqn:bd_s_1}
\end{align}
\normalsize
\footnotesize
The BDepu for $\mathcal{S}_2$ is:
\begin{align}
    \text{BDepu}(\mathcal{S}_2) = &\left[\frac{\Gamma(12)}{\Gamma(12+2201)}\frac{\Gamma(4+885)\Gamma(8+1316)}{\Gamma(4)\Gamma(8)}\right]\nonumber\\
        &\left[ \frac{\Gamma( {4})}{\Gamma( {4+885})}\frac{\Gamma( {4+885})}{\Gamma( {4})}\frac{\Gamma( {8})}{\Gamma( {8+1316})}\frac{\Gamma( {4+1207})\Gamma( {4+109})}{\Gamma( {4})^{2}}
       \right]\nonumber\\
         &\left[\frac{\Gamma( {4})}{\Gamma( {4+885})}\frac{\Gamma( {2+862})\Gamma( {2+23})}{\Gamma( {2})^2}\frac{\Gamma( {8})}{\Gamma( {8+1316})}\frac{\Gamma( {4+869})\Gamma( {4+447})}{\Gamma( {4})^2}\right]\nonumber\\
        &\left[\frac{\Gamma( {2})}{\Gamma( {2+862})}\frac{\Gamma( {1 +192})\Gamma( {1+670})}{\Gamma( {1})^{2}}
        \frac{\Gamma( {2})}{\Gamma( {2+23)}}\frac{\Gamma( {1+20})\Gamma( {1+3})}{\Gamma( {1})^{2}}
        \right]\nonumber \\
        &\left[\frac{\Gamma( {2})}{\Gamma( {2+805})}\frac{\Gamma( {1+146})\Gamma( {1+659})}{\Gamma( {1})^{2}}
         \frac{\Gamma( {2})}{\Gamma( {2+402})}\frac{\Gamma( {1+296})\Gamma( {1+106})}{\Gamma( {1})^{2}}
        \right]\nonumber \\
        &\left[\frac{\Gamma( {4})}{\Gamma( {4+109)}}\frac{\Gamma( {2+57})\Gamma( {2+52})}{\Gamma( {2})^{2}}
        \right]\label{eqn:bd_s_2}
\end{align}
\normalsize
\footnotesize
The BDepu for the resized subtree in $\mathcal{S}_2$ is:
\begin{align}
    &\left[ \frac{\Gamma( {4})}{\Gamma( {4+885})}\frac{\Gamma( {4+885})}{\Gamma( {4})}\frac{\Gamma( {4})}{\Gamma( {4+885})}\frac{\Gamma( {2+862})\Gamma( {2+23})}{\Gamma( {2})^2}\right]\nonumber\\
    &\left[\frac{\Gamma( {2})}{\Gamma( {2+862})}\frac{\Gamma( {1 +192})\Gamma( {1+670})}{\Gamma( {1})^{2}}
        \frac{\Gamma( {2})}{\Gamma( {2+23)}}\frac{\Gamma( {1+20})\Gamma( {1+3})}{\Gamma( {1})^{2}}
        \right]\nonumber \\
        =&\left[ \frac{\Gamma( {4})}{\Gamma( {4+885})}\frac{\Gamma( {1 +192})\Gamma( {1+670})}{\Gamma( {1})^{2}}
        \frac{\Gamma( {1+20})\Gamma( {1+3})}{\Gamma( {1})^{2}}
        \right]\nonumber \\
        =&\left[ \frac{\Gamma( {4})}{\Gamma( {4+885})}\frac{\Gamma( {1 +192})\Gamma( {1+670})\Gamma( {1+20})\Gamma( {1+3})}{\Gamma( {1})^{4}}
        \right] \label{eqn:bd_s_3}
\end{align}
\normalsize
which is equal to the BDepu contribution for the corresponding floret in $\mathcal{S}_3$.

\bibliographystyle{elsarticle-harv} 
\bibliography{biblio}

\begin{thebibliography}{23}
\expandafter\ifx\csname natexlab\endcsname\relax\def\natexlab#1{#1}\fi
\providecommand{\url}[1]{\texttt{#1}}
\providecommand{\href}[2]{#2}
\providecommand{\path}[1]{#1}
\providecommand{\DOIprefix}{doi:}
\providecommand{\ArXivprefix}{arXiv:}
\providecommand{\URLprefix}{URL: }
\providecommand{\Pubmedprefix}{pmid:}
\providecommand{\doi}[1]{\href{http://dx.doi.org/#1}{\path{#1}}}
\providecommand{\Pubmed}[1]{\href{pmid:#1}{\path{#1}}}
\providecommand{\bibinfo}[2]{#2}
\ifx\xfnm\relax \def\xfnm[#1]{\unskip,\space#1}\fi
\bibitem[{Barclay et~al.(2013)Barclay, Hutton and Smith}]{barclay13}
\bibinfo{author}{Barclay, L.M.}, \bibinfo{author}{Hutton, J.L.},
  \bibinfo{author}{Smith, J.Q.}, \bibinfo{year}{2013}.
\newblock \bibinfo{title}{Refining a {B}ayesian network using a chain event
  graph}.
\newblock \bibinfo{journal}{International {J}ournal of {A}pproximate
  {R}easoning} \bibinfo{volume}{54}, \bibinfo{pages}{1300--1309}.
\bibitem[{Barclay et~al.(2014)Barclay, Hutton and Smith}]{barclayhutton14}
\bibinfo{author}{Barclay, L.M.}, \bibinfo{author}{Hutton, J.L.},
  \bibinfo{author}{Smith, J.Q.}, \bibinfo{year}{2014}.
\newblock \bibinfo{title}{Chain event graphs for informed missingness}.
\newblock \bibinfo{journal}{Bayesian {A}nalysis} \bibinfo{volume}{9},
  \bibinfo{pages}{53--76}.
\bibitem[{Carli et~al.(2022)Carli, Leonelli, Riccomagno and
  Varando}]{stagedtrees20}
\bibinfo{author}{Carli, F.}, \bibinfo{author}{Leonelli, M.},
  \bibinfo{author}{Riccomagno, E.}, \bibinfo{author}{Varando, G.},
  \bibinfo{year}{2022}.
\newblock \bibinfo{title}{The r package stagedtrees for structural learning of
  stratified staged trees}.
\newblock \bibinfo{journal}{Journal of Statistical Software}
  \bibinfo{volume}{102}, \bibinfo{pages}{1–--30}.
\bibitem[{Chickering(1995)}]{chickering95}
\bibinfo{author}{Chickering, D.M.}, \bibinfo{year}{1995}.
\newblock \bibinfo{title}{A transformational characterization of equivalent
  bayesian network structures}, in: \bibinfo{booktitle}{Proceedings of the 11th
  Conference on Uncertainty in Artificial Intelligence}, pp.
  \bibinfo{pages}{87--98}.
\bibitem[{Collazo et~al.(2018)Collazo, G{\"o}rgen and Smith}]{collazo18}
\bibinfo{author}{Collazo, R.A.}, \bibinfo{author}{G{\"o}rgen, C.},
  \bibinfo{author}{Smith, J.Q.}, \bibinfo{year}{2018}.
\newblock \bibinfo{title}{Chain event graphs}.
\newblock \bibinfo{publisher}{{CRC} {P}ress}.
\bibitem[{Cooper and Herskovits(1991)}]{cooper1991bayesian}
\bibinfo{author}{Cooper, G.F.}, \bibinfo{author}{Herskovits, E.},
  \bibinfo{year}{1991}.
\newblock \bibinfo{title}{A bayesian method for constructing bayesian belief
  networks from databases}, in: \bibinfo{booktitle}{Uncertainty Proceedings
  1991}. \bibinfo{publisher}{Elsevier}, pp. \bibinfo{pages}{86--94}.
\bibitem[{Cowell and Smith(2014)}]{cowell14}
\bibinfo{author}{Cowell, R.G.}, \bibinfo{author}{Smith, J.Q.},
  \bibinfo{year}{2014}.
\newblock \bibinfo{title}{Causal discovery through {MAP} selection of
  stratified chain event graphs}.
\newblock \bibinfo{journal}{Electronic {J}ournal of {S}tatistics}
  \bibinfo{volume}{8}, \bibinfo{pages}{965--997}.
\bibitem[{Freeman and Smith(2011)}]{freemansmith11}
\bibinfo{author}{Freeman, G.}, \bibinfo{author}{Smith, J.Q.},
  \bibinfo{year}{2011}.
\newblock \bibinfo{title}{Bayesian {MAP} model selection of chain event
  graphs}.
\newblock \bibinfo{journal}{Journal of {M}ultivariate {A}nalysis}
  \bibinfo{volume}{102}, \bibinfo{pages}{1152--1165}.
\bibitem[{G{\"o}rgen et~al.(2022)G{\"o}rgen, Maraj and Nicklasson}]{gorgen21}
\bibinfo{author}{G{\"o}rgen, C.}, \bibinfo{author}{Maraj, A.},
  \bibinfo{author}{Nicklasson, L.}, \bibinfo{year}{2022}.
\newblock \bibinfo{title}{Staged tree models with toric structure}.
\newblock \bibinfo{journal}{Journal of Symbolic Computation}
  \bibinfo{volume}{113}, \bibinfo{pages}{242--268}.
\bibitem[{G{\"o}rgen and Smith(2018)}]{gorgen18}
\bibinfo{author}{G{\"o}rgen, C.}, \bibinfo{author}{Smith, J.Q.},
  \bibinfo{year}{2018}.
\newblock \bibinfo{title}{Equivalence classes of staged trees}.
\newblock \bibinfo{journal}{Bernoulli} \bibinfo{volume}{24},
  \bibinfo{pages}{2676--2692}.
\bibitem[{Heckerman et~al.(1995)Heckerman, Geiger and Chickering}]{heckerman95}
\bibinfo{author}{Heckerman, D.}, \bibinfo{author}{Geiger, D.},
  \bibinfo{author}{Chickering, D.M.}, \bibinfo{year}{1995}.
\newblock \bibinfo{title}{Learning bayesian networks: The combination of
  knowledge and statistical data}.
\newblock \bibinfo{journal}{Machine learning} \bibinfo{volume}{20},
  \bibinfo{pages}{197--243}.
\bibitem[{Koller et~al.(2007)Koller, Friedman, Getoor and
  Taskar}]{koller2007graphical}
\bibinfo{author}{Koller, D.}, \bibinfo{author}{Friedman, N.},
  \bibinfo{author}{Getoor, L.}, \bibinfo{author}{Taskar, B.},
  \bibinfo{year}{2007}.
\newblock \bibinfo{title}{Graphical models in a nutshell}.
\newblock \bibinfo{journal}{Introduction to statistical relational learning}
  \bibinfo{volume}{43}.
\bibitem[{Neapolitan(2003)}]{neap03}
\bibinfo{author}{Neapolitan, R.}, \bibinfo{year}{2003}.
\newblock \bibinfo{title}{Learning Bayesian Networks}.
\newblock \DOIprefix\doi{10.1145/1327942.1327961}.
\bibitem[{Scutari(2016)}]{scutari16}
\bibinfo{author}{Scutari, M.}, \bibinfo{year}{2016}.
\newblock \bibinfo{title}{An empirical-bayes score for discrete bayesian
  networks}, in: \bibinfo{booktitle}{Proceedings of the 8th {I}nternational
  {C}onference on {P}robabilistic {G}raphical {M}odels},
  \bibinfo{organization}{PMLR}. pp. \bibinfo{pages}{438--448}.
\bibitem[{Shafer(1996)}]{shafer96}
\bibinfo{author}{Shafer, G.}, \bibinfo{year}{1996}.
\newblock \bibinfo{title}{The art of causal conjecture}.
\newblock \bibinfo{publisher}{{MIT} {P}ress}.
\bibitem[{Shenvi(2021)}]{shenvi21}
\bibinfo{author}{Shenvi, A.}, \bibinfo{year}{2021}.
\newblock \bibinfo{title}{Non-Stratified Chain Event Graphs: Dynamic Variants,
  Inference and Applications}.
\newblock Ph.D. thesis. University of Warwick.
\bibitem[{Shenvi and Smith(2020)}]{shenvi2020constructing}
\bibinfo{author}{Shenvi, A.}, \bibinfo{author}{Smith, J.Q.},
  \bibinfo{year}{2020}.
\newblock \bibinfo{title}{Constructing a chain event graph from a staged tree},
  in: \bibinfo{booktitle}{Proceedings of the 10th {I}nternational {C}onference
  on {P}robabilistic {G}raphical {M}odels}, \bibinfo{organization}{PMLR}. pp.
  \bibinfo{pages}{437--448}.
\bibitem[{Shenvi et~al.(2018)Shenvi, Smith, Walton and
  Eldridge}]{shenvi2018modelling}
\bibinfo{author}{Shenvi, A.}, \bibinfo{author}{Smith, J.Q.},
  \bibinfo{author}{Walton, R.}, \bibinfo{author}{Eldridge, S.},
  \bibinfo{year}{2018}.
\newblock \bibinfo{title}{Modelling with non-stratified chain event graphs},
  in: \bibinfo{booktitle}{International Conference on Bayesian Statistics in
  Action}, pp. \bibinfo{pages}{155--163}.
\bibitem[{Silander et~al.(2007)Silander, Kontkanen and
  Myllym\"{a}ki}]{silander07sensi}
\bibinfo{author}{Silander, T.}, \bibinfo{author}{Kontkanen, P.},
  \bibinfo{author}{Myllym\"{a}ki, P.}, \bibinfo{year}{2007}.
\newblock \bibinfo{title}{On sensitivity of the map bayesian network structure
  to the equivalent sample size parameter}, in: \bibinfo{booktitle}{Proceedings
  of the Twenty-Third Conference on Uncertainty in Artificial Intelligence},
  \bibinfo{publisher}{AUAI Press}, \bibinfo{address}{Arlington, Virginia, USA}.
  p. \bibinfo{pages}{360–367}.
\bibitem[{Smith and Anderson(2008)}]{smith08}
\bibinfo{author}{Smith, J.Q.}, \bibinfo{author}{Anderson, P.E.},
  \bibinfo{year}{2008}.
\newblock \bibinfo{title}{Conditional independence and chain event graphs}.
\newblock \bibinfo{journal}{Artificial {I}ntelligence} \bibinfo{volume}{172},
  \bibinfo{pages}{42--68}.
\bibitem[{Thwaites et~al.(2010)Thwaites, Smith and Riccomagno}]{thwaites10}
\bibinfo{author}{Thwaites, P.A.}, \bibinfo{author}{Smith, J.Q.},
  \bibinfo{author}{Riccomagno, E.}, \bibinfo{year}{2010}.
\newblock \bibinfo{title}{Causal analysis with chain event graphs}.
\newblock \bibinfo{journal}{Artificial {I}ntelligence} \bibinfo{volume}{174},
  \bibinfo{pages}{889--909}.
\bibitem[{Yu and Smith(2021)}]{yu2021causal}
\bibinfo{author}{Yu, X.}, \bibinfo{author}{Smith, J.Q.}, \bibinfo{year}{2021}.
\newblock \bibinfo{title}{Causal algebras on chain event graphs with informed
  missingness for system failure}.
\newblock \bibinfo{journal}{Entropy} \bibinfo{volume}{23},
  \bibinfo{pages}{1308}.
\bibitem[{Zhang and Poole(1999)}]{zhang1999role}
\bibinfo{author}{Zhang, N.L.}, \bibinfo{author}{Poole, D.},
  \bibinfo{year}{1999}.
\newblock \bibinfo{title}{On the role of context-specific independence in
  probabilistic inference}, in: \bibinfo{booktitle}{Proceedings of the 16th
  {I}nternational {J}oint {C}onference on {A}rtificial {I}ntelligence}, pp.
  \bibinfo{pages}{1288--1293}.

\end{thebibliography}






\end{document}